\PassOptionsToPackage{final}{graphicx}

\documentclass[nolimits,a4paper,final]{llncs}

\usepackage[colorlinks,linkcolor={blue},citecolor={blue},urlcolor={red},breaklinks=true,final]{hyperref}

\usepackage[utf8]{inputenc}

\usepackage{graphbox}

\usepackage{booktabs} %

\def\case{}

\providecommand{\catname}{\mathbf} 
\providecommand{\clsname}{\mathcal}
\providecommand{\oname}[1]{\operatorname{\mathsf{#1}}}

\def\defcatname#1{\expandafter\def\csname B#1\endcsname{\catname{#1}}}
\def\defcatnames#1{\ifx#1\defcatnames\else\defcatname#1\expandafter\defcatnames\fi}
\defcatnames ABCDEFGHIJKLMNOPQRSTUVWXYZ\defcatnames

\def\defclsname#1{\expandafter\def\csname C#1\endcsname{\clsname{#1}}}
\def\defclsnames#1{\ifx#1\defclsnames\else\defclsname#1\expandafter\defclsnames\fi}
\defclsnames ABCDEFGHIJKLMNOPQRSTUVWXYZ\defclsnames

\def\defbbname#1{\expandafter\def\csname BB#1\endcsname{\mathbb{#1}}}
\def\defbbnames#1{\ifx#1\defbbnames\else\defbbname#1\expandafter\defbbnames\fi}
\defbbnames ABCDEFGHIJKLMNOPQRSTUVWXYZ\defbbnames

\def\Set{\catname{Set}}
\def\Cpo{\catname{Cpo}}

\providecommand{\argument}{\operatorname{-\!-}}

\providecommand{\mplus}{{\scriptscriptstyle\bf+}} 	%

\providecommand{\FSet}{{\mathcal P}_{\omega}}		%
\providecommand{\Id}{\operatorname{Id}}

\providecommand{\Hom}{\mathsf{Hom}}
\providecommand{\id}{\mathsf{id}}
\providecommand{\comp}{\mathbin{\circ}}

\providecommand{\tensor}{\mathbin{\otimes}}
\providecommand{\unit}{\star}				
\providecommand{\bang}{\operatorname!}				%

\providecommand{\dar}{\kern-1.2pt\operatorname{\downarrow}}	
\providecommand{\uar}{\kern-1.2pt\operatorname{\uparrow}}	
\providecommand{\mto}{\mapsto}
\providecommand{\xto}[1]{\xrightarrow{#1}}

\providecommand{\fst}{\oname{fst}}
\providecommand{\snd}{\oname{snd}}
\providecommand{\pr}{\oname{pr}}

\providecommand{\brks}[1]{\langle #1\rangle}

\providecommand{\inl}{\oname{inl}}
\providecommand{\inr}{\oname{inr}}
\providecommand{\inj}{\oname{in}}

\DeclareSymbolFont{Symbols}{OMS}{cmsy}{m}{n}
\DeclareMathSymbol{\iobj}{\mathord}{Symbols}{"3B}
\providecommand{\curry}{\oname{curry}}

\providecommand{\ev}{\oname{ev}}

\providecommand{\comma}{,\operatorname{}\linebreak[1]}		%
\providecommand{\dash}{\nobreakdash-\hspace{0pt}}			%

\providecommand{\by}[1]{\text{/\!/~#1}}			%
\providecommand{\pacman}[1]{}					%

\providecommand{\mone}{{\text{\kern.5pt\rmfamily-}\sf\kern-.5pt1}}

\makeatletter
\@ifpackageloaded{enumitem}{}{\usepackage[loadonly]{enumitem}}		%
\makeatother

\newlist{citemize}{itemize}{1}
\setlist[citemize]{label=\labelitemi,wide} %

\newlist{cenumerate}{enumerate}{1}
\setlist[cenumerate,1]{label=\arabic*.~,ref={\arabic*},wide} %

\makeatletter
\def\mfix#1{\oname{#1}\@ifnextchar\bgroup\@mfix{}}	%
\def\@mfix#1{#1\@ifnextchar\bgroup\mfix{}}			%
\makeatother

\providecommand{\case}[3]{\mfix{case}{\mathbin{}#1}{of}{#2}{\kern-1pt;}{\mathbin{}#3}}

\spnewtheorem{thm}[theorem]{Theorem}{\bfseries}{\itshape}
\spnewtheorem{cor}[theorem]{Corollary}{\bfseries}{\itshape}
\spnewtheorem{cnj}[theorem]{Conjecture}{\bfseries}{\itshape}
\spnewtheorem{lem}[theorem]{Lemma}{\bfseries}{\itshape}
\spnewtheorem{lemdefn}[theorem]{Lemma and Definition}{\bfseries}{\itshape}
\spnewtheorem{prop}[theorem]{Proposition}{\bfseries}{\itshape}
\spnewtheorem{defn}[theorem]{Definition}{\bfseries}{\upshape}
\spnewtheorem{rem}[theorem]{Remark}{\bfseries}{\upshape}
\spnewtheorem{notation}[theorem]{Notation}{\bfseries}{\upshape}
\spnewtheorem{expl}[theorem]{Example}{\bfseries}{\upshape}
\spnewtheorem{thmdefn}[theorem]{Theorem and Definition}{\bfseries}{\itshape}
\spnewtheorem{propdefn}[theorem]{Proposition and Definition}{\bfseries}{\itshape}
\spnewtheorem{assumption}[theorem]{Assumption}{\bfseries}{\upshape}
\spnewtheorem{algorithm}[theorem]{Algorithm}{\bfseries}{\upshape}

 \renewenvironment{corollary}{\begin{cor}}{\end{cor}}

\usepackage{microtype}
\usepackage{needspace}

\usepackage{todos}

\renewcommand{\labelitemi}{$\vcenter{\hbox{\rule{1.2ex}{1pt}}}$}

\usepackage{savesym}

\savesymbol{degree}
\savesymbol{leftmoon}
\savesymbol{rightmoon}
\savesymbol{fullmoon}
\savesymbol{newmoon}
\savesymbol{diameter}
\savesymbol{emptyset}
\savesymbol{bigtimes}

 \ifdraft
 {
 \usepackage[layout=footnote,draft]{fixme}
 \usepackage[notcite,notref]{showkeys}
 
 }{
\usepackage[layout=footnote,final]{fixme}
}

\FXRegisterAuthor{sg}{asg}{SG}	%
\FXRegisterAuthor{ls}{als}{LS}	%

\usepackage{etex}
\usepackage{amsmath,amssymb} 

\usepackage{proof}

\usepackage{subcaption}
\usepackage[all]{xy}
\usepackage{fontawesome}

\usepackage{textcomp}

\usepackage{bm}
\usepackage[matha]{mathabx}
\restoresymbol{other}{emptyset}

\usepackage{tikz}
\usepackage{tkz-euclide}
\usetkzobj{all}
\usepackage{tikz-cd}

\tikzset{
commutative diagrams/.cd,
arrow style=tikz,
diagrams={>=stealth},
row sep=large, 
column sep = huge
}

\usepackage{pict2e}

\newcommand{\algebra}{a}

\usepackage{wrapfig}

\renewcommand{\comp}{\kern1pt}

\renewcommand{\inl}{\inj_1}
\renewcommand{\inr}{\inj_2}

\renewcommand{\fst}{{\pr_1}}
\renewcommand{\snd}{{\pr_2}}

\newcommand{\cpto}{
  \mathrel{\raisebox{0.5ex}{\kern1pt\ensuremath{\mathrel{\tikz{ \draw [-stealth,line width=0.4] (0.4ex,.9ex) -- (0,.9ex) -- (0,0.4ex) -- (2.5ex,0.4ex); }}}\kern1pt}}
}

\newcommand{\klstar}{\star}	%
\newcommand{\istar}{\dagger}  			%
\newcommand{\iistar}{\ddagger}  		%
\newcommand{\sarr}{\mathbin{\kern-1pt\text{\rotatebox[origin=c]{45}{$\downarrow$}}}}

\newcommand{\out}{\oname{out}}
\newcommand{\ini}{\oname{in}}

\newcommand{\tr}{\oname{tr}}

\newcommand{\grd}{\operatorname{\blacktriangleright}}

\newcommand{\GHom}{\Hom^{\kern-1pt\bullet}}
\newcommand{\IHom}{\Hom^{\kern-.2pt\scalebox{.58}{$\grd$}}}

\newcommand{\HS}{\oname{HS}}
\newcommand{\N}{\oname{N}}

\newcommand{\real}{\bm{\mathsf{R}}}

\newcommand{\lrule}[3]{{#1}~\frac{#2}{#3}}

\def\Hilb{\catname{Hilb}}

\newcommand{\norm}[1]{|\!| #1|\!|}

\renewcommand{\paragraph}[1]{\par\indent\textbf{#1.}}

\makeatletter
\def\moverlay{\mathpalette\mov@rlay}
\def\mov@rlay#1#2{\leavevmode\vtop{%
   \baselineskip\z@skip \lineskiplimit-\maxdimen
   \ialign{\hfil$\m@th#1##$\hfil\cr#2\crcr}}}
\newcommand{\charfusion}[3][\mathord]{
    #1{\ifx#1\mathop\vphantom{#2}\fi
        \mathpalette\mov@rlay{#2\cr#3}
      }
    \ifx#1\mathop\expandafter\displaylimits\fi}
\makeatother

\newcommand{\cupdot}{\charfusion[\mathbin]{\cup}{\cdot}}

\allowdisplaybreaks

\begin{document}

\title{Guarded Traced Categories}  

\author{Sergey Goncharov\and Lutz Schröder}
\institute{Friedrich-Alexander-Universit\"at Erlangen-N{\"u}rnberg}
\maketitle

\begin{abstract}
  Notions of guardedness serve to delineate the admissibility of
  cycles, e.g.\ in recursion, corecursion, iteration, or tracing. We
  introduce an abstract notion of guardedness structure on a symmetric
  monoidal category, along with a corresponding notion of guarded
  traces, which are defined only if the cycles they induce are
  guarded. We relate structural guardedness, determined by propagating
  guardedness along the operations of the category, to geometric
  guardedness phrased in terms of a diagrammatic language. In our
  setup, the Cartesian case (recursion) and the co-Cartesian case
  (iteration) become completely dual, and we show that in these cases,
  guarded tracedness is equivalent to presence of a guarded Conway
  operator, in analogy to an observation on total traces by Hasegawa and
  Hyland. Moreover, we relate guarded traces to unguarded categorical
  uniform fixpoint operators in the style of Simpson and
  Plotkin. Finally, we show that partial traces based on
  Hilbert-Schmidt operators in the category of Hilbert spaces are an
  instance of guarded traces.
\end{abstract}

\section{Introduction}
In models of computation, various notions of \emph{guardedness} serve
to control cyclic behaviour by allowing only guarded cycles, with the
aim to ensure properties such as solvability of recursive equations or
productivity. Typical examples are guarded process algebra
specifications~\cite{Milner89,BaetenBastenEtAl10}, coalgebraic guarded
(co-)recursion~\cite{Rutten00,Milius05}, finite delay in online Turing
machines~\cite{BookGreibach70}, and productive definitions in
intensional type theory~\cite{AbelPientka13,Mgelberg14}, but also
contractive maps in (ultra-)metric spaces~\cite{KrishnaswamiBenton11}.

A highly general model for unrestricted cyclic computations, on the
\begin{wrapfigure}{r}{0.3\textwidth}
\vspace{-20pt}
  \begin{center}
  \includegraphics[align=c,scale=.3]{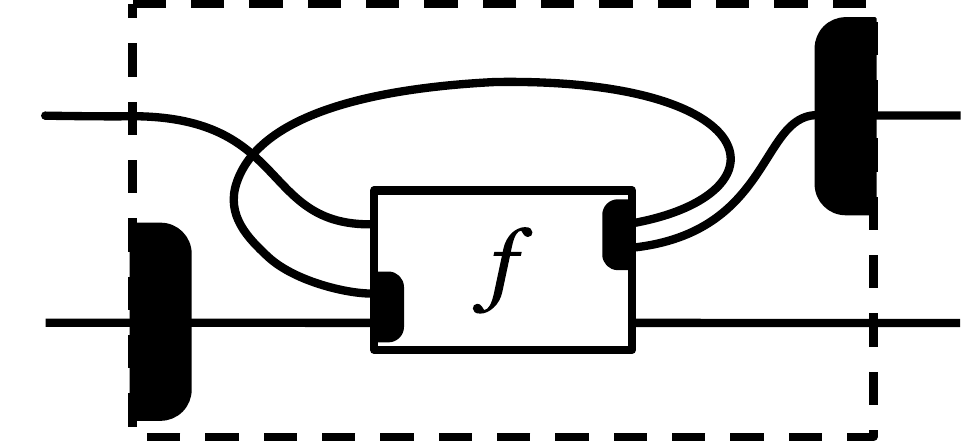}
  \end{center}
\vspace{-10pt}
  \caption{Guarded trace}\label{fig:guarded-trace}
\vspace{-20pt}
\end{wrapfigure}
other hand, are \emph{traced monoidal
  categories}~\cite{JoyalStreetEtAl96}; besides \emph{recursion} and
\emph{iteration}, they cover further kinds of cyclic behaviour, e.g.\
in Girard's \emph{Geometry of
  Interaction}~\cite{Girard89,AbramskyHaghverdiEtAl02} and quantum
programming~\cite{AbramskyCoecke04,Selinger04}.  In the present paper
we parametrize the framework of traced symmetric monoidal categories
with a notion of guardedness, arriving at \emph{(abstractly) guarded
  traced categories}, which effectively vary between two extreme
cases: symmetric monoidal categories (nothing is guarded) and traced
symmetric monoidal categories (everything is guarded). In terms of the
standard diagrammatic language for traced monoidal categories, we
decorate input and output gates of boxes to indicate guardedness; the
diagram governing trace formation would then have the general form
depicted in Figure~\ref{fig:guarded-trace}
-- that is, we can only form traces connecting guarded (black) output
gates to input gates that are unguarded (black), i.e.\ not assumed to
be already guarded.

We provide basic structural results on our notion of abstract
guardedness, and identify a wide array of examples. Specifically, we
establish a geometric characterization of guardedness in terms of
paths in diagrams; we identify a notion of \emph{guarded ideal}, along
with a construction of guardedness structures from guarded ideals and
simplifications of this construction for the (co-)Cartesian and the
Cartesian closed case; and we describe `vacuous' guardedness
structures where traces do not actually generate proper diagrammatic cycles. In
terms of examples, we begin with the case where the monoidal structure
is either product (Cartesian), corresponding to guarded recursion, or
coproduct (co-Cartesian), for guarded iteration; the axioms for
guardedness allow for a basic duality that indeed makes these two
cases precisely dual. For total traces in Cartesian categories,
Hasegawa and Hyland observed that trace operators are in one-to-one
correspondence with \emph{Conway fixpoint
  operators}~\cite{Hasegawa97,Hasegawa99}; we extend this
correspondence to the guarded case, showing that guarded trace
operators on a Cartesian category are in one-to-one correspondence
with guarded Conway operators. In a more specific setting, we relate
\emph{guarded} traces in Cartesian categories to \emph{unguarded}
categorical uniform fixpoints as studied by Crole and
Pitts~\cite{CrolePitts90} and by Simpson and
Plotkin~\cite{Simpson92,SimpsonPlotkin00}. Concluding with a case
where the monoidal structure is a proper tensor product, we show that
the partial trace operation on (infinite-dimentional) Hilbert spaces is an instance of
vacuous guardedness; this result relates to work by Abramsky, Blute,
and Panangaden on traces over nuclear ideals, in this case over
\emph{Hilbert-Schmidt operators}~\cite{AbramskyBluteEtAl99}.\medskip

\noindent\textbf{Related work}\quad Abstract guardedness serves to determine
definedness of a guarded trace operation, and thus relates to work on
partial traces. We discuss work on nuclear
ideals~\cite{AbramskyBluteEtAl99} in Section~\ref{sec:dag}. In
\emph{partial traced
  categories}~\cite{HaghverdiScott10,MalherbeScottEtAl12}, traces are
governed by a partial equational version (consisting of both strong
and directed equations) of the Joyal-Street-Verity axioms; morphisms
for which trace is defined are called \emph{trace class}. A key
difference to the approach via guardedness is that being trace class
applies only to morphisms with inputs and outputs of matching types
while guardedness applies to arbitrary morphisms, allowing for
compositional propagation. Also, the axiomatizations are incomparable:
Unlike for trace class morphisms~\cite[Remark 2.2]{HaghverdiScott10},
we require guardedness to be closed under composition with arbitrary
morphisms (thus covering contractivity but not, e.g., monotonicity as
in the modal $\mu$-calculus); on the other hand, as noted by
Jeffrey~\cite{Jeffrey12}, guarded traces, e.g.\ of contractions, need
not satisfy Vanishing II as a Kleene equality as assumed in partial
traced categories.  Some approaches treat traces as partial over
objects~\cite{BluteCockettEtAl00,Jeffrey97}.  In concrete algebraic
categories, partial traces can be seen as induced by total traces in
an ambient category of relations~\cite{ArthanMartinEtAl09}. We discuss
work on guardedness via endofunctors in Remark~\ref{rem:ml}.

\section{Preliminaries}\label{sec:prelim}
We recall requisite categorical notions; see~\cite{MacLane71} for a
comprehensive introduction.\medskip

\noindent\textbf{Symmetric Monoidal Categories}\quad
A \emph{symmetric monoidal category} $(\BC,\tensor, I)$ consists of a
category $\BC$ (with object class $|\BC|$), a bifunctor $\tensor$
(\emph{tensor product}), and a \emph{(tensor) unit} $I\in |\BC|$, and
coherent isomorphisms witnessing that $\tensor$ is, up to isomorphism,
a commutative monoid structure with unit~$I$. For the latter, we
reserve the notation
$\alpha_{A,B,C}:(A\tensor B)\tensor C\cong A\tensor (B\tensor C)$
(\emph{associator}), $\gamma_{A,B}: A\tensor B\cong B\tensor A$
(\emph{symmetry}), and $\upsilon_A:I\tensor A\cong A$ (\emph{left
  unitor}); the \emph{right unitor} $\hat\upsilon_A:A\tensor I\cong A$
is expressible via the symmetry.
A symmetric monoidal category is \emph{Cartesian} if the monoidal
structure is finite product (i.e.\ $\tensor =\times$, and $I=1$ is a
terminal object), and, dually, \emph{co-Cartesian} if the monoidal
structure is finite coproduct (i.e.\ $\tensor=+$, and $I=\iobj$ is an
initial object). Coproduct injections are written
$\inj_i:X_i\to X_1+X_2$ ($i=1,2$), and product projections
$\pr_i:X_1\times X_2\to X_i$. Various notions of algebraic tensor
products also induce symmetric monoidal structures; see
Section~\ref{sec:dag} for the case of Hilbert spaces. One has an
obvious expression language for objects and morphisms in symmetric
monoidal categories~\cite{Selinger11}, the former obtained by
postulating basic objects and closing under $I$ and $\tensor$, and the
latter by postulating basic morphisms of given profile and closing
under $\tensor$, $I$, composition, identities, and the monoidal
isomorphisms, subject to the evident notion of
\emph{well-typedness}. Morphism expressions are conveniently
represented as \emph{diagrams} consisting of boxes representing the
basic morphisms, with input and output gates corresponding to the
given profile. Tensoring is represented by putting boxes on top of
each other, and composition by wires connecting outputs to
inputs~\cite{Selinger11}. In a \emph{traced symmetric monoidal
  category} one has an additional operation (\emph{trace}) that
essentially enables the formation of loops in diagrams, as in
Figure~\ref{fig:guarded-trace} (but without decorations).

\textbf{Monads and (Co-)algebras}\quad A(n \emph{$F$-)coalgebra} for a
functor $F:\BC\to\BC$ is a pair $(X,f:X\to FX)$ where $X\in |\BC|$,
thought of as modelling states and generalized
transitions~\cite{Rutten00}.  A \emph{final coalgebra} is a final
object in the category of coalgebras (with $\BC$-morphisms $h:X\to Y$ such
that $(Fh) f = g h$ as morphisms $(X,f)\to (Y,g)$), denoted
$(\nu F,\out:\nu F\to F\nu F)$ if it exists.
Dually, an \emph{$F$-algebra} has the form $(X,f:FX\to X)$. %
A \emph{monad} $\BBT=(T,\mu,\eta)$ on a category $\BC$ consists of an
endofunctor $T$ on $\BC$ and natural transformations $\eta:\Id\to T$
(\emph{unit}) and $\mu:T^2\to T$ (\emph{multiplication}) subject to
standard equations~\cite{MacLane71}. As observed by
Moggi~\cite{Moggi91a}, monads can be seen as capturing
\emph{computational effects} of programs, with $TX$ read as a type of
computations with side effects from $T$ and results in~$X$. %
In this view, the \emph{Kleisli category} $\BC_\BBT$ of $\BBT$, which
has the same objects as $\BC$ and
$\Hom_{\BC_\BBT}(X,Y)= \Hom_{\BC}(X,TY)$, is a category of
side-effecting programs. A monad is \emph{strong} if it is equipped
with a \emph{strength}, i.e.\ a natural transformation
$X\times TY\to T(X\times Y)$ satisfying evident coherence conditions
(e.g.~\cite{Moggi91a}). A $T$-algebra $(A,\algebra)$ is an
\emph{(Eilenberg-Moore) $\BBT$-algebra} (for the \emph{monad}~$\BBT$)
if additionally $\algebra\comp\eta=\id$ and
$\algebra\comp (T\algebra) = \algebra\mu_A$; the category of
$\BBT$-algebras is denoted $\BC^{\BBT}$.
\section{Guarded Categories}\label{sec:guarded}
\vspace{-0.09em}
\noindent We now introduce our notion of guarded structure. A standard
example of guardedness are guarded definitions in process
algebra. E.g.\ in the definition $P = a.P$, the right hand occurrence
of $P$ is guarded, ensuring unique solvability (by a process that
keeps outputting $a$). A further example is contractivity of maps
between complete metric spaces. We formulate abstract closure
properties for \emph{partial} guardedness where only some of the
inputs and outputs of a morphism are guarded. Specifically, we
distinguish \emph{guarded outputs} and \emph{guarded inputs} ($D$
and~$B$, respectively, in the following definition), with the intended
reading that guarded outputs yield guarded data \emph{provided}
guarded data is already provided at guarded inputs, while unguarded
inputs may be fed arbitrarily.
\begin{figure}[t]
\begin{center}%
\small%
\includegraphics[scale=0.23]{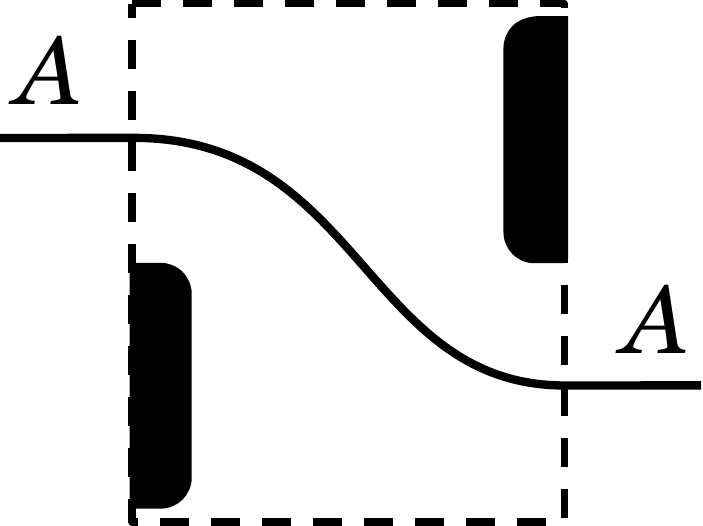}
\qquad
\includegraphics[scale=0.23]{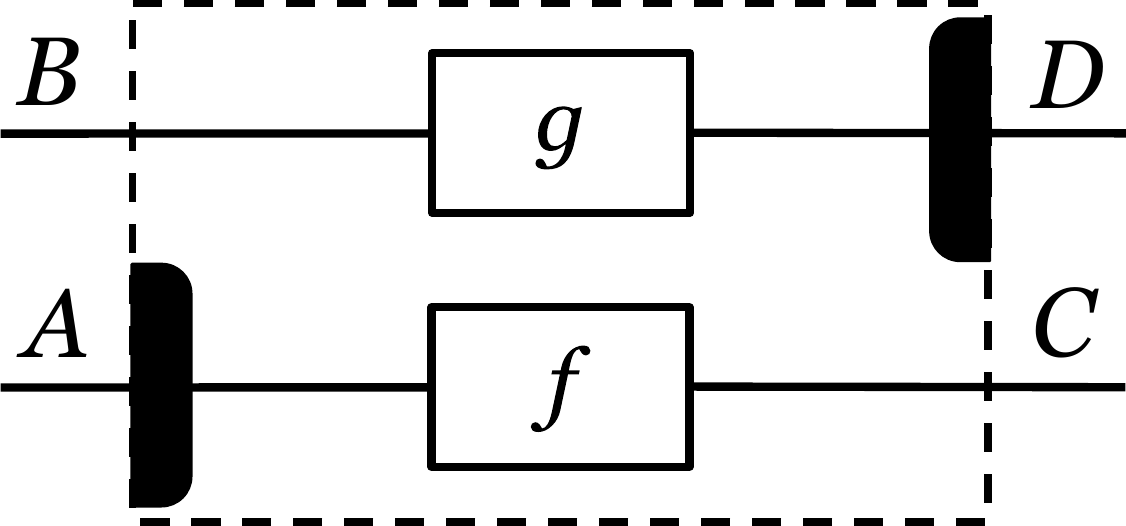}
\qquad
\includegraphics[scale=0.23]{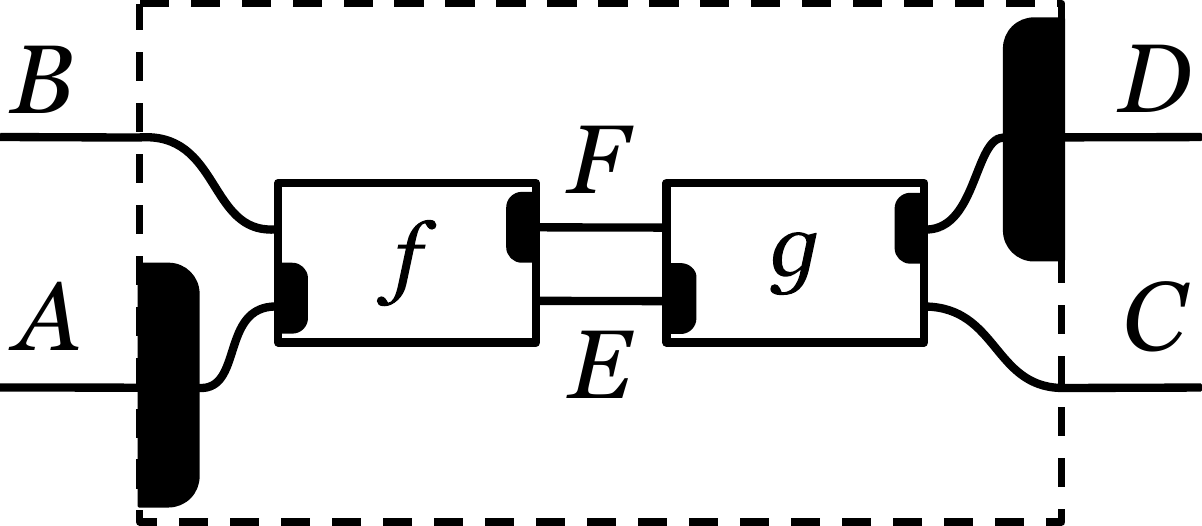}
\qquad
\includegraphics[scale=0.23]{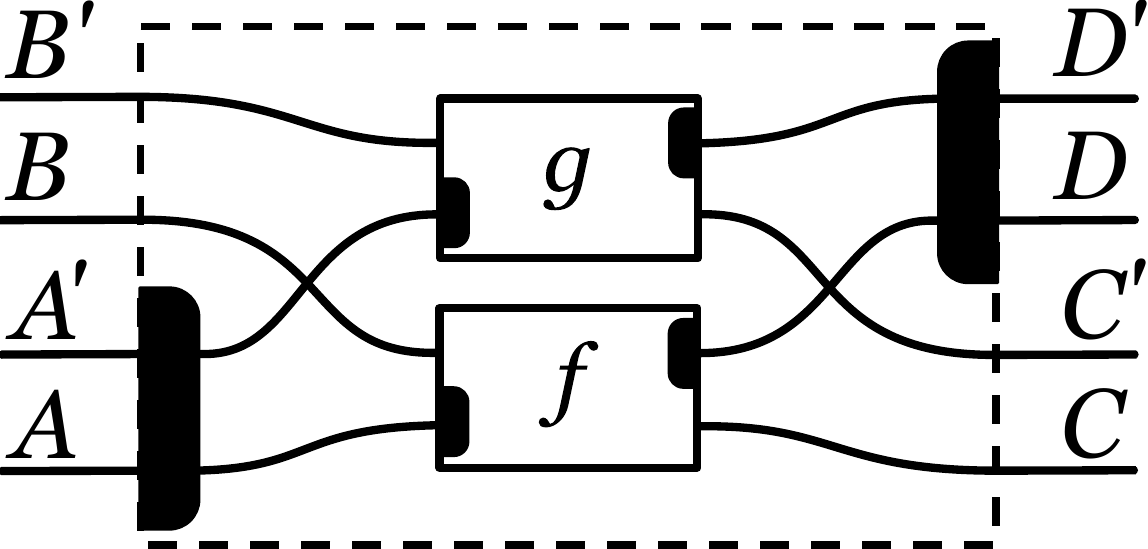}
\caption{Axioms of guarded categories}
\label{fig:gmon}
\end{center}
\vspace{-5ex}
\end{figure}
\begin{defn}[Guarded category]
\label{def:guard_sm}
An \emph{(abstractly) guarded category} is a symmetric monoidal
category $(\BC, \tensor, I)$ equipped with distinguished subsets
$\GHom(A\tensor B,C\tensor D)\subseteq\Hom({A\tensor B},C\tensor D)$
of \emph{partially guarded morphisms} for $A,B,C,D\in|\BC|$,
satisfying the following conditions:
\begin{description}
 \item[(uni${}_{\tensor}$)] $\gamma_{I,A}\in\GHom(I\tensor A, A\tensor I)$; 
 \item[(vac${}_{\tensor}$)] $f\tensor g\in \GHom(A\tensor B, C\tensor D)$ for all $f:A\to C$, $g:B\to D$;
 \item[(cmp${}_{\tensor}$)] $g\in\GHom(A\tensor B, E\tensor F)$ and\/ 
$f\in\GHom(E\tensor F, C\tensor D)$ imply  $f\comp g\in \GHom(A\tensor B, C\tensor D)$;
\item[(par${}_{\tensor}$)] for $f\in\GHom(A\tensor B, C\tensor D)$,
  $g\in\GHom(A'\tensor B', C'\tensor D')$, the evident transpose of
  $f\tensor g%
  $ %
is in $\GHom((A\tensor A')\tensor (B\tensor B'), (C\tensor C')\tensor (D\tensor D'))$.
\end{description}
We emphasize that $\GHom(A\tensor B,C\tensor D)$ is meant to depend
individually on $A$, $B$, $C$, $D$ and not just on $A\tensor B$ and
$C\tensor D$.
\end{defn}
\noindent 
\begin{wrapfigure}{r}{0.3\textwidth}
\vspace{-30pt}
  \begin{center}
  \includegraphics[align=c,scale=.3]{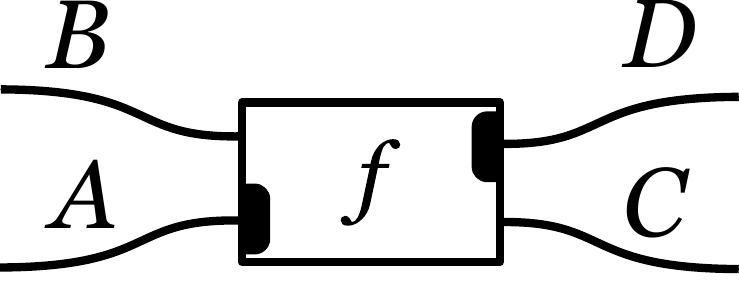}
  \end{center}
\vspace{-30pt}
\end{wrapfigure}
One easily derives a \emph{weakening} rule stating that if
$f\in\GHom((A\tensor A')\tensor B,C\tensor (D'\tensor D))$, then the
obvious transpose of $f$ is in
$\GHom(A\tensor (A'\tensor B),(C\tensor D')\tensor D)$.

We extend the standard diagram language for symmetric monoidal
categories (Section~\ref{sec:prelim}), representing morphisms
$f\in\GHom(A\tensor B,C\tensor D)$ by \emph{decorated boxes} as shown
on the right, with black bars marking the \emph{unguarded input} gates
$A$ and the \emph{guarded output} gates $D$. Weakening then
corresponds to shrinking the black bars of decorated boxes.
Figure~\ref{fig:gmon} depicts the above axioms in this language.
Solid boxes represent the assumptions, while dashed boxes represent
the conclusions. The latter only occur in the derivation process and
do not form part of the actual diagrams representing concrete
morphisms. We silently identify object expressions and sets of gates
in diagrams. Given a (well-typed) morphism expression $e$, a judgement
$e\in\GHom(A\tensor B,C\tensor D)$, called a \emph{guardedness typing}
of~$e$, is \emph{derivable} if it can be derived from the assumed
guardedness typing of the constituent basic boxes of $e$ using the
rules in Definition~\ref{def:guard_sm}.  We have an obvious notion of
(directed) \emph{paths} in diagrams; a path is
\emph{guarded}\label{def:guarded-path} if it passes some basic box~$f$
through an unguarded input gate and a guarded output gate
(intuitively, guardedness is then introduced along the path as the
passage through~$f$ will guarantee guarded output without assuming
guarded input).  We then have the following geometric characterization
of guardedness typing:
\begin{thm}\label{thm:gcompl}
  For a well-typed morphism expression
  $e\in\Hom(A\tensor B,C\tensor D)$, the guardedness typing
  $e\in\GHom(A\tensor B,C\tensor D)$ is derivable iff in the diagram
  of~$e$, every path from an input gate in $A$ to an output gate in
  $D$ is guarded.
\end{thm}
\noindent Every symmetric monoidal category has both a largest
($\GHom(A\tensor B,C\tensor D)=\Hom(A\tensor B,C\tensor D)$) and a
least guarded structure:
\begin{lemdefn}[Vacuous guardedness]\label{lem:triv}
  Every symmetric mono\-idal category is guarded under taking
  $f\in\GHom(A\tensor B,C\tensor D)$ iff $f$ factors~as
\begin{align*}
A\tensor B \xto{\id_A\tensor g} A\tensor E\tensor D \xto{h\tensor\id_D} C\tensor D
\end{align*}
(eliding associativity) with $g:B\to E\tensor D$, $h:A \tensor E\to C$. This is the least
guarded structure on~$\BC$, the \emph{vacuous guarded structure}.
\end{lemdefn}
E.g.\ the natural guarded structure on Hilbert spaces
(Section~\ref{sec:dag}) is vacuous.
\begin{rem}[Duality]\label{rem:dual}
  The rules and axioms in Figure~\ref{fig:gmon} are stable under
  $180^\degree$\dash rotation, that is, under reversing arrows and
  applying the monoidal symmetry on both sides (this motivates
  decorating the \emph{unguarded} inputs). Consequently, if $\BC$ is
  guarded, then so is the dual category $\BC^{op}$, with guardedness
  given by $f\in\GHom_{\BC^{op}}(A\tensor B, C\tensor D)$ iff the
  obvious transpose of $f$ is in $\GHom_{\BC}(D\tensor C,B\tensor A)$.
\end{rem}
\noindent In case $\tensor$ is coproduct, we can simplify the
description of partial guardedness:
\begin{prop}\label{prop:guard_equiv}
  Partial guardedness in a co-Cartesian category\/ $(\BC, +,\iobj)$ is
  equivalently determined by distinguished subsets
  $\Hom_{\sigma}(X,Y)\subseteq\Hom(X,Y)$ with $\sigma$ ranging over
  coproduct injections $Y_2\to Y_1+Y_2\cong Y$, subject to the rules
  on the right hand side of Figure~\ref{fig:co-cart-g}, where
  $f:X\to_\sigma Y$ denotes $f\in\Hom_\sigma(X,Y)$,
with $f\in\GHom(X_1+X_2,Y_1+Y_2)$ iff $(f\inl)\in\Hom_{\inr}(X_1,Y_1+Y_2)$.
\end{prop}
\begin{figure}[t!]
\begin{center}%
\small%
\begin{flalign*}
&
\lrule{\textbf{(vac${}_{\times}$)}}{f:X\to Z}{f\pr_1:X\times Y\to^{\pr_2} Z} && 
\lrule{\textbf{(vac${}_{\mplus}$)}}{f:X\to Z}{\inl f:X\to_{\inr} Z+Y}\\[2ex]
&
\lrule{\textbf{(cmp${}_{\times}$)}}{
\begin{array}{rl}
f:&X\times Y\to^{\pr_2} Z\\[1ex]
g:&V\to^{\sigma} X\qquad  h:V\to Y
\end{array}}
{f\comp\brks{g,h}:V\to^{\sigma} Z}&&
\lrule{\textbf{(cmp${}_{\mplus}$)}}{
\begin{array}{rl}
f:&X\to_{\inr} Y+Z\\[1ex]
g:&Y\to_{\sigma}V\qquad  h:Z\to V
\end{array}}
{[g,h]\comp f:X\to_{\sigma} V}
\\[2ex]
& 
\lrule{\textbf{(par${}_{\times}$)}}{f:X\to^{\sigma} Y\qquad g:X\to^{\sigma} Z}{\brks{f,g}:X\to^\sigma Y\times Z}&&
\lrule{\textbf{(par${}_{\mplus}$)}}{f:X\to_{\sigma} Z\qquad f:Y\to_{\sigma} Z}{[f,g]:X+Y\to_{\sigma} Z}
\end{flalign*}
\caption{Axioms of Cartesian (left) and co-Cartesian (right) guarded categories}
\label{fig:co-cart-g}
\end{center}
\vspace{-5ex}
\end{figure}
We have used the mentioned rules for $\to_\sigma$ in previous work on
guarded iteration~\cite{GoncharovSchroderEtAl17} (with
\textbf{(vac${}_{\times}$)} called \textbf{(trv)}, and together with
weakening, which as indicated above turns out to be derivable).
By duality (Remark~\ref{rem:dual}), we immediately have a corresponding
description for the Cartesian case:
\begin{corollary}\label{cor:guard_equiv}   
  Partial guardedness in a Cartesian category\/ $(\BC,\times,1)$ is
  equivalently determined by distinguished subsets
  $\Hom^{\sigma}(X,Y)\subseteq\Hom(X,Y)$ with $\sigma$ ranging over
  product projections $X\cong X_1\times X_2 \to X_1$, subject to the
  rules on the left hand side of Figure~\ref{fig:co-cart-g}, where
  $f:X\to^\sigma Y$ denotes $f\in\Hom^\sigma(X,Y)$, with
  $f\in\GHom(X_1\times X_2, Y_1\times Y_2)$ iff
  $\pr_2 f\in\Hom^{\pr_1}(X_1\times X_2 ,Y_2)$.
\end{corollary}
\begin{rem}\label{rem:triv-cocartesian}
  In a co-Cartesian category, vacuous guardedness
  (Lemma~\ref{lem:triv}) can equivalently be described by
  $f\in\GHom(A+B,C+D)$ iff $f$ decomposes as $f=[\inl h,g]$ (uniquely
  provided that $\inl$ is monic), or in terms of the description from
  Proposition~\ref{prop:guard_equiv}, $u\in\Hom_{\inr}(X,Y+Z)$ iff $u$
  factors through $\inl$.  Of course, the dual situation obtains in
  Cartesian categories.
\end{rem}
\begin{expl}[Process algebra]\label{expl:pa}
  Fix a monad $\BBT$ on $(\BC,+,\iobj)$ and an endofunctor
  $\Sigma:\BC\to\BC$ such that the generalized coalgebraic resumption
  transform $T_{\Sigma}= \nu\gamma.\ T(\argument+\Sigma\gamma)$
  exists; we think of $T_\Sigma X$ as a type of processes that have
  side-effects in $\BBT$ and perform communication actions from
  $\Sigma$, seen as a generalized signature. The Kleisli category
  $\BC_{\BBT_{\Sigma}}$ of $\BBT_{\Sigma}$ is again
  co-Cartesian. Putting
\begin{equation*}
f:X\to_{\inr} T_\Sigma (Y+Z) \iff \out f\in\{T(\inl+\id)\comp g\mid g:X\to T(Y+\Sigma T_{\Sigma} (Y+Z))\}
\end{equation*}
(cf.\ Section~\ref{sec:prelim} for notation), we make
$\BC_{\BBT_{\Sigma}}$ into a guarded
category~\cite{GoncharovSchroderEtAl17}.  The standard motivating
example of finitely nondeterministic processes is obtained by taking
$\BBT=\FSet$ (finite powerset monad) and $\Sigma = A\times\argument$
(action prefixing).
\end{expl}

\begin{expl}[Metric spaces]\label{expl:m-space}
  Let $\BC$ be the Cartesian category of metric spaces and
  non-expansive maps.  Taking $f:X\times Y\to^{\pr_2} Z$ iff
  $\lambda y.\, f(x,y)$ is contractive for every $x\in X$ makes $\BC$
  into a guarded Cartesian category.
\end{expl}

\section{Guardedness via Guarded Ideals}\label{sec:ideals}
Most of the time, the structure of a guarded category is determined by
morphisms with only unguarded inputs and guarded outputs, which form
an \emph{ideal}:
\begin{defn}[Guarded morphisms]
  A morphism $f:X\to Y$ in a guarded category is \emph{guarded} (as
  opposed to only partially guarded) if
  $\upsilon^\mone_{Y}\comp f\comp\hat\upsilon_{X}\in\GHom(X\tensor
  I\comma I\tensor Y)$;
  we write $\IHom(X,Y)$ for the set of guarded morphisms $f:X\to Y$.
\end{defn}
\begin{defn}[Guarded ideal]
  A family~$G$ of subsets $G(X,Y)\subseteq\Hom(X,Y)$ ($X,Y\in |\BC|$)
  in a monoidal category $(\BC,\tensor,I)$ is a \emph{guarded ideal}
  if it is closed under $\tensor$ and under composition with arbitrary
  $\BC$-morphisms on both sides, and $G(I,I)=\Hom(I,I)$.
\end{defn}
There is always a \emph{least guarded ideal},
$G(X,Y) = \{g\comp f\mid f:X\to I, g:I\to Y\}$. Moreover, as indicated
above:

\begin{lemdefn}\label{lem:total}
  In a guarded category, the sets $\IHom(X,Y)$ form a guarded ideal,
  the guarded ideal \emph{induced} by the guarded structure.
\end{lemdefn}

\noindent Conversely, it is clear that every guarded ideal
\emph{generates} a guarded structure by just closing under the rules
of Definition~\ref{def:guard_sm}.
\begin{defn}[Ideally guarded category]\label{def:ideally-guarded}
  A guarded category is \emph{ideal} or \emph{ideally guarded} (over
  $G$) if it is generated by some guarded ideal~($G$).
 \end{defn}
 \noindent We give a more concrete description:
\begin{thm}\label{thm:t_to_p}
  Let $(\BC,\tensor,I)$ be ideally guarded over $G$. Then
 $\GHom(A\tensor B,C\tensor D)$ consists of the morphisms of the form
\begin{equation*}%
\vcenter{\hbox{\includegraphics[scale=.3]{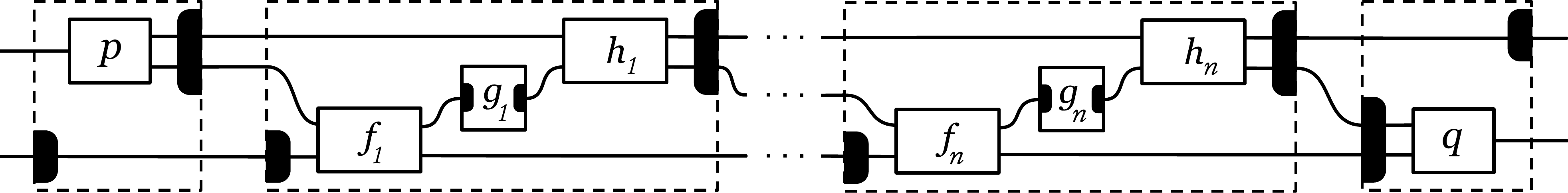}}}
\medskip
\end{equation*}
for $g_i$ in $G$ and arbitrary $p$, $q$, $f_i$, $h_i$.
\end{thm}
The transitions between guarded ideals and guarded structures are not
in general mutually inverse: The guarded structure generated the
guarded ideal induced by a guarded structure may be smaller than the
original one (Example~\ref{expl:non-ideal}), and the guarded ideal
induced by the guarded structure generated by a guarded ideal $G$ may
be larger than~$G$ (Remark~\ref{rem:triv-ideal}). We proceed to
analyse details.
\begin{prop}\label{prop:triv-ideal}
  On every symmetric monoidal category, the least guarded structure
  (Lemma~\ref{lem:triv}) is ideal.
\end{prop}
\begin{rem}\label{rem:triv-ideal}
  Vacuously guarded categories need not induce the least guarded ideal
  (although by the next results, this does hold in the Cartesian and
  the co-Cartesian case). In fact, by Lemma~\ref{lem:triv}, the
  guarded ideal induced by the vacuous guarded structure consists of
  the morphisms of the form $(h\tensor\id_D)(\id_A\tensor g)$ (eliding
  associativity and the unitor) where $g:I\to E\tensor D$,
  $h:A \tensor E\to I$:
  \begin{equation}\label{eq:trivial-ideal}
    \includegraphics[align=c,scale=.3]{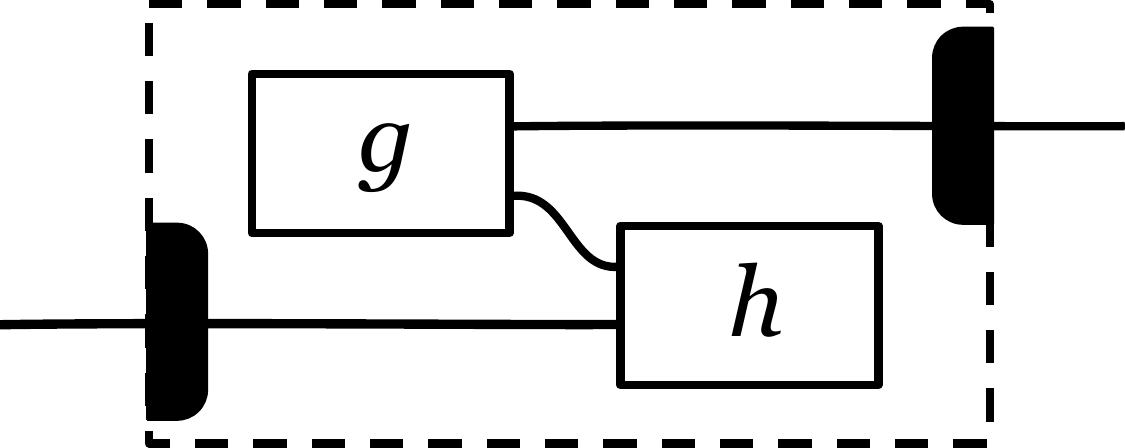}
  \end{equation} 
  This ideal will resurface in the discussion of Hilbert spaces
  (Section~\ref{sec:dag}).
\end{rem}
\noindent The situation is simpler in the Cartesian and, dually, in the
co-Cartesian case.
\begin{lem}\label{lem:induce} Let $\BC$ be ideally guarded over $G$, and
  suppose that every $f\in G({X\tensor Y}\comma Z)$ factors through
  $\hat f\tensor\id:X\tensor Y\to V\tensor Y$ for some $\hat f\in G(X,V)$. Then 
the guardedness structure of $\BC$ induces~$G$.
\end{lem}
If $\tensor=+$, the premise of the lemma is automatic, since
$f\in G(X+Y,Z)$ can be represented as
$[f\inl,f\inr] = [\id,f\inr]\comp (f\inl+\id)$ where
$f\inl\in G(X,Z)$ by the closure properties of guarded ideals. Hence, we
obtain
\begin{thm}\label{thm:cart-total}
  The guarded structure generated by a guarded ideal $G$ on a
  co-Cartesian category is equivalently described by
  $\Hom_{\inr}(X,Y+Z)= \{[\inl,g] h\mid g\in G(W, Y+Z), h:X\to Y+
  W\}$, and hence induces $G$.
\end{thm}
\begin{corollary}\label{cor:cart-total1}
  The guarded structure generated by a guarded ideal $G$ on a Cartesian
  category is equivalently described by
  $\Hom^{\fst}(X\times Y,Z)= \{h\comp \brks{g, \snd}\mid g\in
  G(X\times Y,W), h:W\times Y\to Z\}$, and hence induces $G$.
\end{corollary}
The description can be further simplified in the Cartesian closed
case.
\begin{corollary}\label{cor:cart-total2}
  Given a guarded ideal $G$ on a Cartesian closed category, put
  $f:X\times Y\to^{\pr_1} Z$ iff\/ $\curry f\in G(X,Z^Y)$. This
  describes the guarded structure induced by $G$ iff $G$ is
  \emph{exponential}, i.e.\ $f\in G(X,Y)$ implies $f^V\in G(X^V,Y^V)$.
\end{corollary}
(We leave it as an open question whether a similar characterization
holds in the monoidal closed case.) Natural examples of both ideal and
non-ideal guardedness are found in metric spaces:
\begin{expl}[Metric spaces]\label{expl:non-ideal}
  The guarded structure on metric spaces from
  Example~\ref{expl:m-space} fails to be ideal: It induces the guarded
  ideal of contractive maps, which however generates the (ideal)
  guarded structure described by $f:X\times Y\to^{\snd} Z$ iff
  $f(x,y)$ is \emph{uniformly} contractive in $y$, i.e.\ there is
  $c<1$ such that for every $x$, $\lambda y.\,f(x,y)$ is contractive
  with contraction factor $c$.
\end{expl}
\noindent A large class of ideally guarded structures arises as
follows.
\begin{prop}\label{prop:mult-guard}
Let\/ $\BC$ be a Cartesian category equipped with an endofunctor \mbox{$\grd:\BC\to\BC$} and 
a natural transformation $\oname{next}:\Id\to\grd$. Then the following definition
yields a guarded ideal in $\BC$: 
$G(X,Y) = \{f\oname{next}\mid f:\grd X\to Y\}$.
The arising guarded structure is
$\Hom^{\fst}(X\times Y,Z) = \{f\brks{\oname{next},\snd}\mid f:\grd
(X\times Y)\times Y\to Z\}$.
If moreover $\oname{next}:X\times Y\to\grd(X\times Y)$ factors through
$\oname{next}\times\id:X\times Y\to\grd X\times Y$, then
$\Hom^{\fst}(X\times Y,Z) = \{f\comp(\oname{next}\times\id)\mid f:\grd
X\times Y\to Z\}$.
\end{prop}

\begin{rem}\label{rem:ml}
  Proposition~\ref{prop:mult-guard} connects our approach to previous
  work based precisely on the assumptions of the
  proposition~\cite{MiliusLitak17} (in fact, the term guarded traced
  category is already used there, with different meaning). A
  limitation of the approach via a functor $\grd$ arises from the
  need to fix $\grd$ globally, so that, e.g., the ideal guarded
  structure on metric spaces (Example~\ref{expl:non-ideal}) is not
  covered -- capturing contractivity via $\grd$ requires fixing a
  single global contraction factor.
\end{rem}
\noindent The following instance of Proposition~\ref{prop:mult-guard}
has received extensive recent interest in programming semantics:
\begin{expl}[Topos of Trees]\label{expl:tt}
  Let $\BC$ be the \emph{topos of
    trees}~\cite{BirkedalMgelbergEtAl12}, i.e.\ the presheaf category
  $\Set^{\omega^{op}}$ where~$\omega$ is the preorder of natural
  numbers (starting from $1$) ordered by inclusion. An object $X$ of
  $\BC$ is thus a family $(X(n))_{n=1,2\ldots}$ of sets with
  restriction maps $r_n:X(n+1)\to X(n)$. The \emph{later}-endofunctor
  $\grd:\BC\to\BC$ is defined by $\grd X(1) = \{\unit\}$ and
  $\grd X(n+1) = X(n)$, and the natural transformation
  $\oname{next}_X:X\to\,\grd X$ by
  $\oname{next}_X (1)=\bang:X(1)\to \{\star\}$,
  $\oname{next}_X (n+1)=r_{n+1}:X{(n+1})\to X(n)$.  Guarded morphisms
  according to Proposition~\ref{prop:mult-guard} are called \emph{contractive},
  generalizing the metric setup.  Contractive morphisms form an
  exponential ideal, so partial guardedness is described as in
  Corollary~\ref{cor:cart-total2}, and hence agrees with contractivity
  in part of the input as in
  \cite[Definition~2.2]{BirkedalMgelbergEtAl12}.

\end{expl}

\section{Guarded Traces}%

As indicated previously, the main purpose of our notion of abstract
guardedness is to enable fine-grained control over the formation of
feedback loops, viz, \emph{traces}.
\begin{figure}[t!]
\begin{center}%
	\begin{subfigure}[t]{0.35\textwidth}
        \flushleft  
		\includegraphics[scale=0.2]{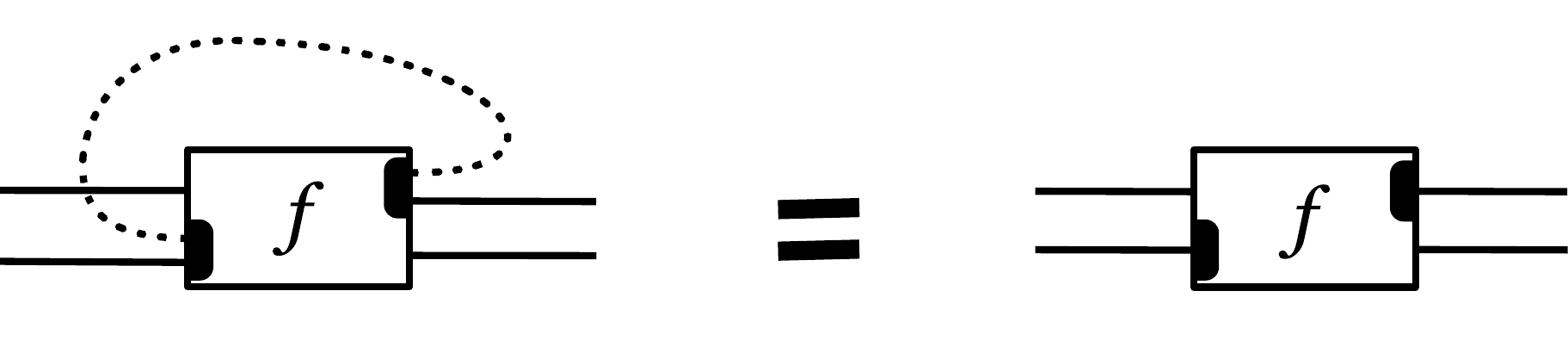}
        \caption{Vanishing I}
    \end{subfigure}%
\qquad
	\begin{subfigure}[t]{0.55\textwidth}
        \flushright
		\includegraphics[scale=0.2]{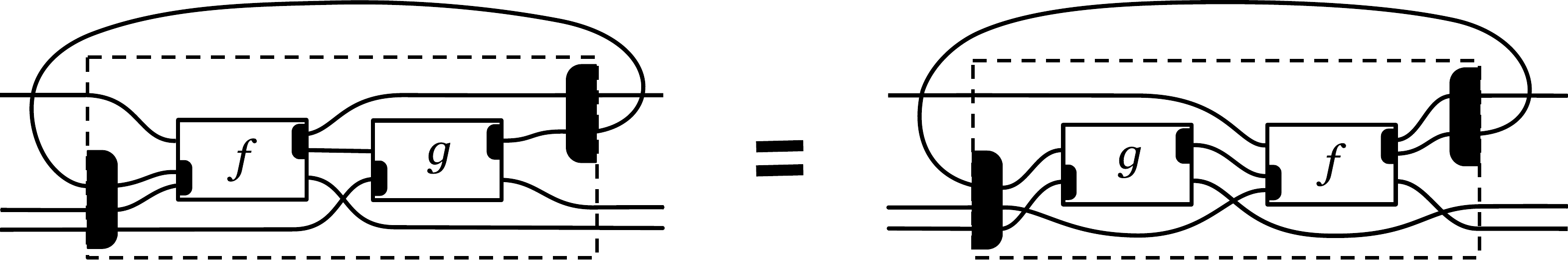}
        \caption{Sliding I}
    \end{subfigure}%
\\[1em]
	\begin{subfigure}[t]{0.4\textwidth}
        \flushleft
		\includegraphics[scale=0.2]{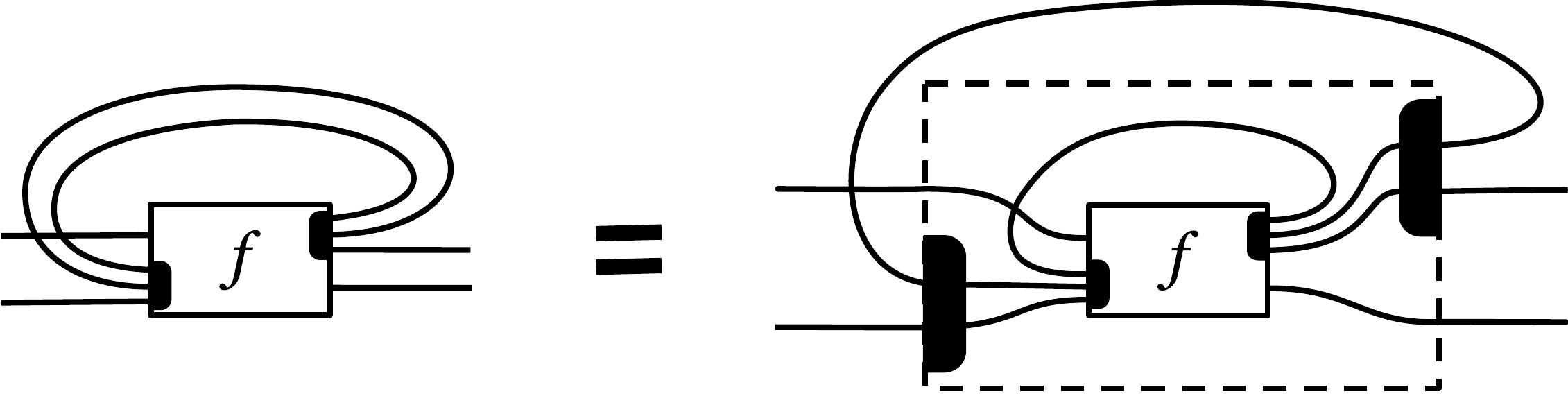}
        \caption{Vanishing II}
    \end{subfigure}%
\quad
	\begin{subfigure}[t]{0.55\textwidth}
        \flushright
		\includegraphics[scale=0.2]{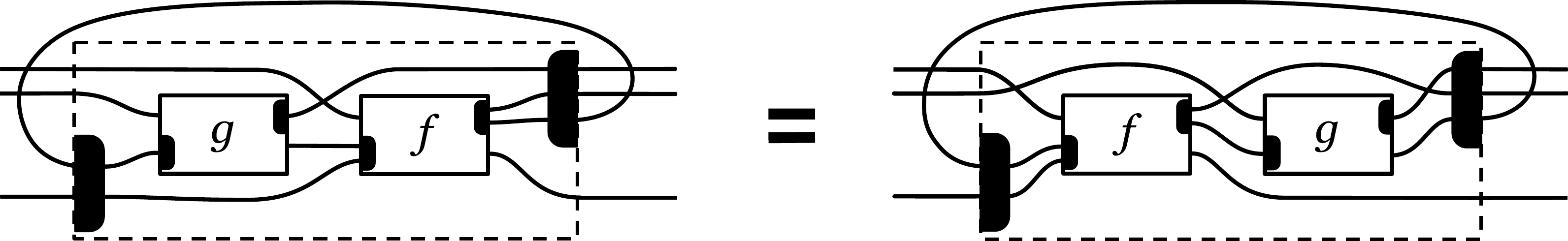}
        \caption{Sliding II}
    \end{subfigure}%
\\[1em]
	\begin{subfigure}[t]{0.41\textwidth}
        \flushleft
		\includegraphics[scale=0.2]{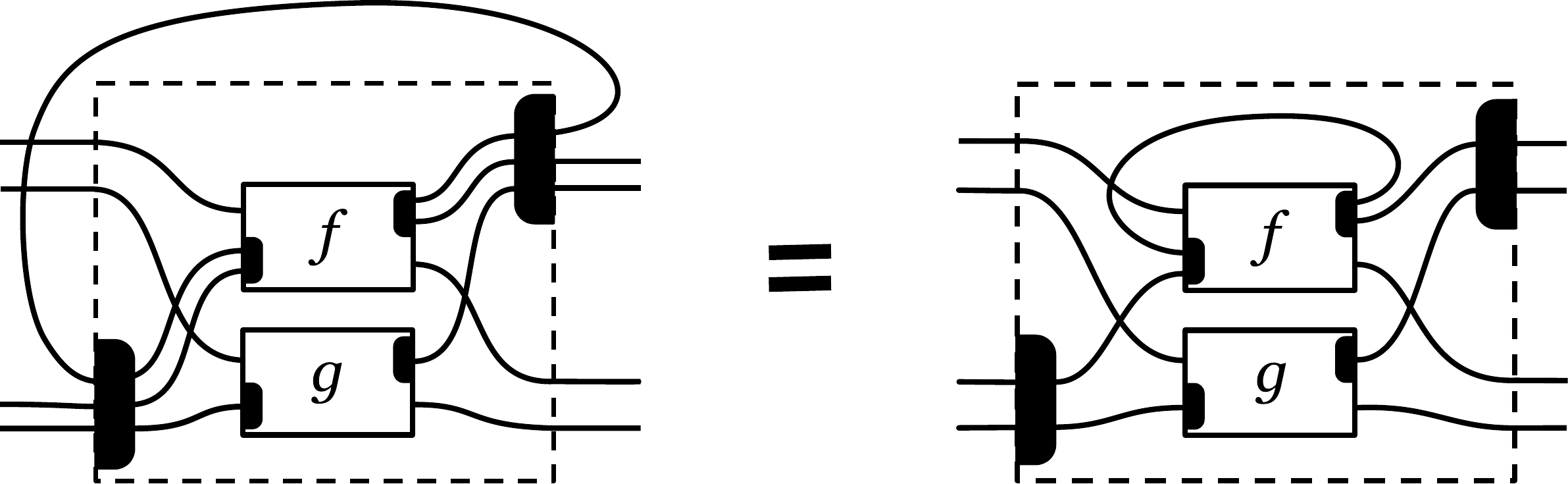}
        \caption{Superposing}
    \end{subfigure}%
\quad
	\begin{subfigure}[t]{0.53\textwidth}
        \flushright
		\includegraphics[scale=0.2]{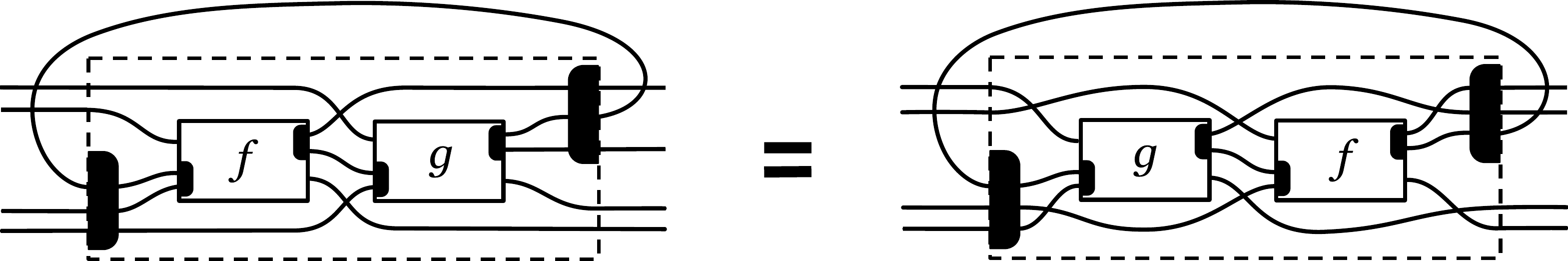}
        \caption{Sliding III}
    \end{subfigure}%
\\[1em]
\begin{subfigure}[t]{0.58\textwidth}
        \centering
\includegraphics[scale=0.2]{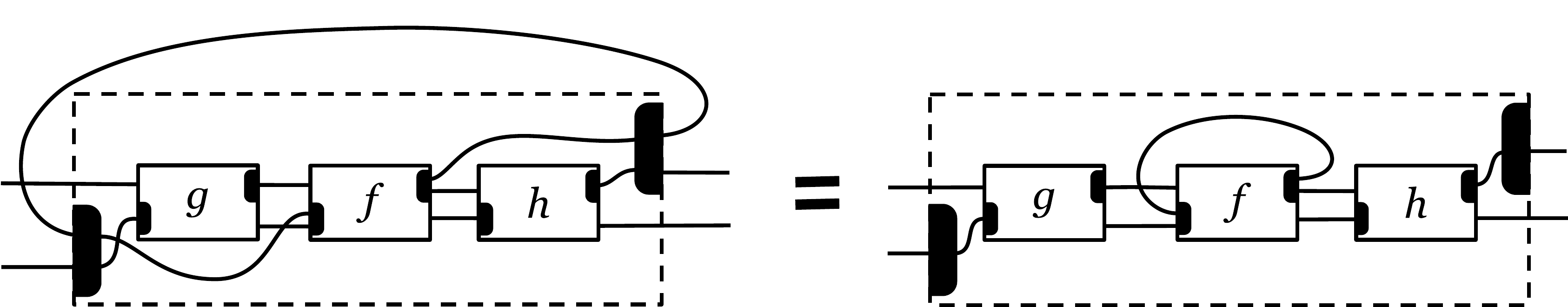}
        \caption{Tightening}
    \end{subfigure}%
\qquad
\begin{subfigure}[t]{0.35\textwidth}
        \centering
		\includegraphics[scale=0.2]{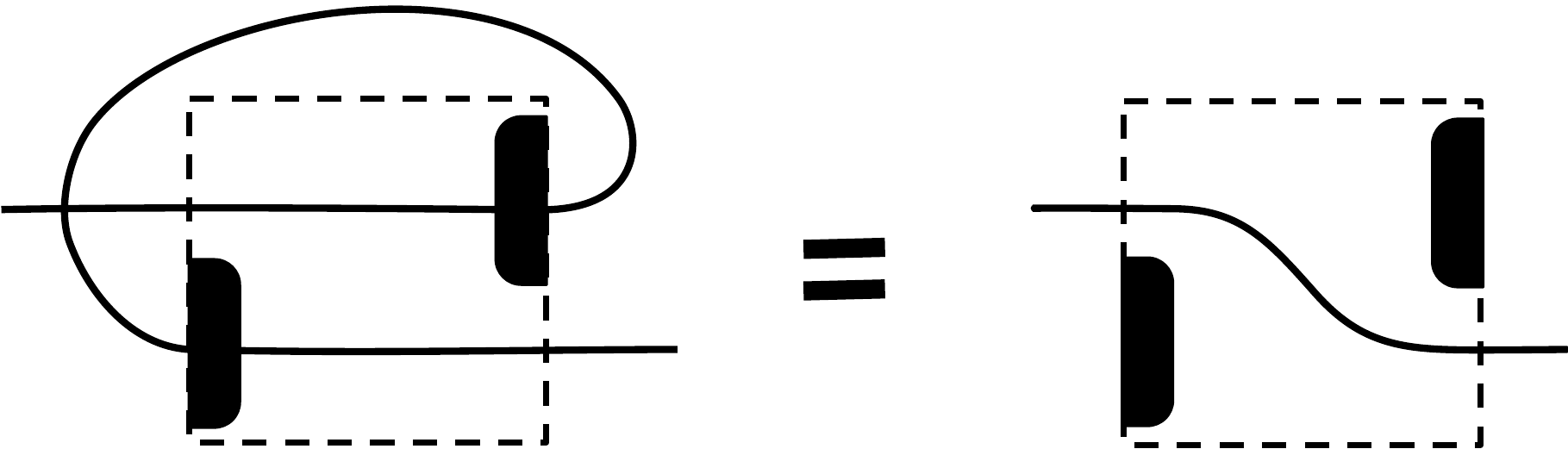}
        \caption{Yanking}
    \end{subfigure}%
\\[1em]
\caption{Axioms of guarded traced categories}
\label{fig:gtrace}
\end{center}
\vspace{-5ex}
\end{figure}%
\begin{defn}[Guarded traced category]\label{def:guarded-trace}
We call a guarded category $(\BC,\tensor,I)$ \emph{guarded traced} if it is 
equipped with a \emph{guarded trace operator}
\begin{displaymath}
\tr_{A,B,C,D}^U:\GHom((A\tensor U)\tensor B,C\tensor (D\tensor U)) \to \GHom(A\tensor B,C\tensor D),
\end{displaymath}
visually corresponding to the diagram formation rule in Figure~\ref{fig:guarded-trace},
so that the adaptation of the Joyal-Street-Verity axiomatization of
traced symmetric monoidal categories~\cite{JoyalStreetEtAl96} shown in
Figure~\ref{fig:gtrace} is satisfied.
\end{defn}
\begin{rem}
  The versions of the sliding axiom in Figure~\ref{fig:gtrace} differ in
  the way the loop is guarded. They are in line with duality
  (Remark~\ref{rem:dual}): Sliding II arises from Sliding I by
  $180^\degree$ rotation, and Sliding III is symmetric under
  $180^\degree$ rotation.
\end{rem}
We proceed to investigate the geometric properties of guarded traced
categories, partly extending Theorem~\ref{thm:gcompl}. The
syntactic setting extends the one for guarded categories by
additionally closing morphism expressions under the trace operator
(interpreted diagrammatically as in Figure~\ref{fig:guarded-trace}),
obtaining \emph{traced morphism expressions}. Term formation thus
becomes mutually recursive with guardedness typing: if $e$ is a traced
morphism expression such that
$e\in\GHom((A\tensor U)\tensor B,C\tensor (D\tensor U))$ is derivable,
then $\tr_{A,B,C,D}(e)$ is a traced morphism expression, and
$\tr_{A,B,C,D}(e)\in\GHom(A\tensor B,C\tensor D)$ is derivable.
\emph{Traced diagrams} consists of finitely many (decorated) basic
boxes and wires connecting output gates of basic boxes to input gates,
with each gate attached to at most one wire; open gates are regarded
as inputs or outputs, respectively, of the whole diagram. Of course,
acyclicity is not required. %
We first note that the easy direction of
Theorem~\ref{thm:gcompl} adapts straightforwardly to the setting
with traces:
\begin{prop}\label{prop:ggcompl}
  Let $e$ be a traced morphism expression such that
  $e\in\GHom(A\tensor B,C\tensor D)$ is derivable. Then in the diagram
  of $e$,
  all loops and all paths from input gates in $A$ to
  output gates in $D$ are guarded (p.~\pageref{def:guarded-path}).
\end{prop}
\noindent Remarkably, the converse of Proposition~\ref{prop:ggcompl}
in general fails in several ways:
\begin{expl}\label{expl:n-idl-counter}
  The left diagram below
  \begin{equation}\label{eq:ctr-expls}
    \includegraphics[align=c,scale=.25]{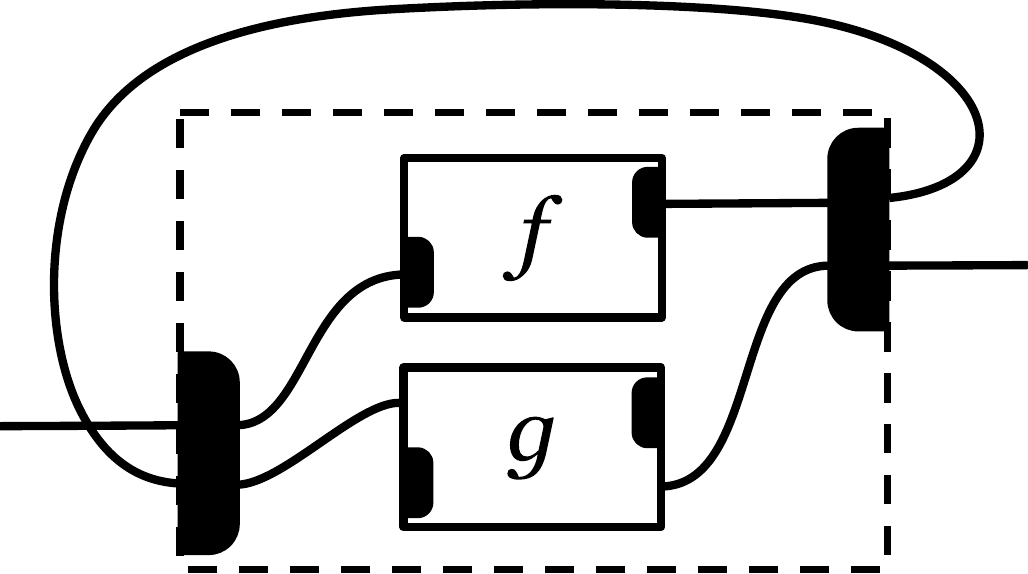}\qquad\qquad
    \includegraphics[align=c,scale=.25]{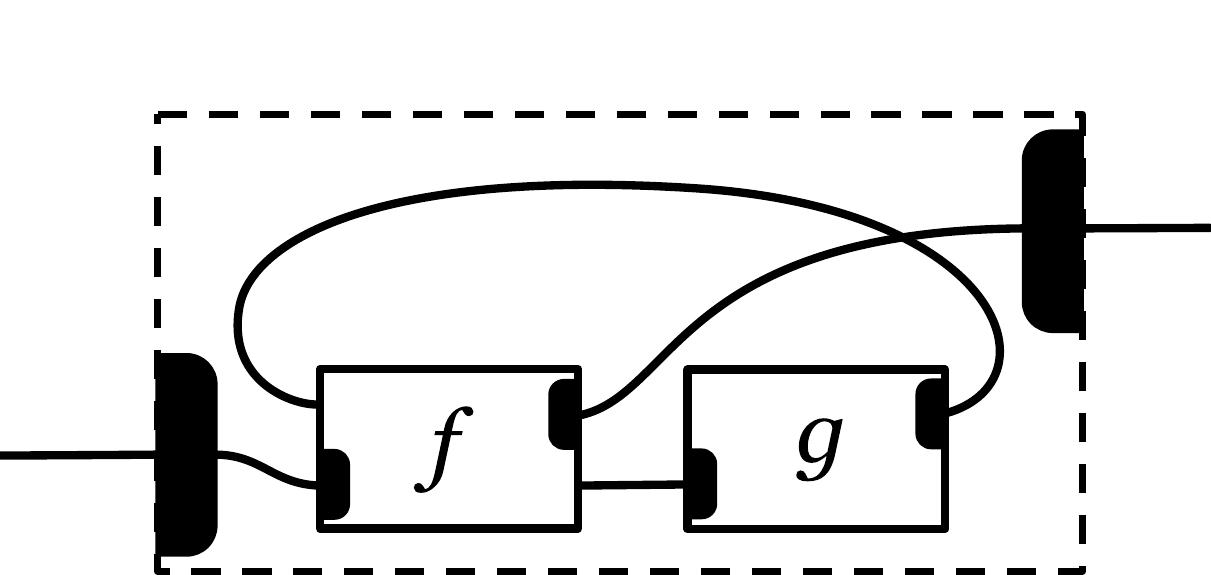}
  \end{equation}
  shows that guardedness typing is not closed under equality of traced
  morphism expressions: Write $e$ for the expression inducing the
  dashed box. By Proposition~\ref{prop:ggcompl},~$e$, and hence
  $\tr(e)$, fail to type as indicated. However, $\tr(e)=gf$, for which
  the overall guardedness typing indicated is easily derivable.

  Moreover, the diagram on the right above satisfies the necessary
  condition from Proposition~\ref{prop:ggcompl} but is not induced by
  an expression for which the indicated guardedness typing is
  derivable, essentially because both ways of cutting the loop violate
  the necessary condition from Proposition~\ref{prop:ggcompl}.
\end{expl}
\noindent However, if $\BC$ is ideally guarded over a guarded
ideal $G$, we do have a converse to Proposition~\ref{prop:ggcompl}: By
Theorem~\ref{thm:t_to_p}, we can then restrict basic boxes in diagrams
to be either \emph{guarded}, i.e.\ have only black gates, or
\emph{unguarded}, i.e.\ have only white gates. %
We call
the correspondingly restricted diagrams \emph{ideally guarded}. (We
emphasize that the guardedness typing of \emph{composite} ideally
guarded diagrams still needs to mix guarded and unguarded inputs and
outputs.) A path in an ideally guarded diagram is guarded iff it passes
through a guarded basic box.

The left-hand diagram in~\eqref{eq:ctr-expls} is in fact ideally
guarded, so guardedness typing fails to be closed under equality also
in the ideally guarded case. However, for ideally guarded diagrams we
have the following converse of Proposition~\ref{prop:ggcompl}.
\begin{thm}\label{thm:ggcompl-conv}
  Let $\Delta$ be an ideally guarded diagram, with sets of input and
  output gates disjointly decomposed as $A\cupdot B$ and $C\cupdot D$,
  respectively. If every loop in $\Delta$ and every path from a gate
  in $A$ to a gate in $D$ is guarded, then $\Delta$ is induced by a
  traced morphism expression $e$ such that
  $e\in\GHom(A\tensor B,C\tensor D)$ is derivable.
\end{thm}
We next take a look at the Cartesian and co-Cartesian cases.  Recall
that by Proposition~\ref{prop:guard_equiv}, the definition of guarded
category can be simplified if $\tensor=+$ (and dually if
$\tensor=\times$). This simplification extends to guarded traced
categories by generalizing Hyland-Hasegawa's equivalence between
Cartesian trace operators and Conway fixpoint
operators~\cite{Hasegawa97,Hasegawa99}.
\begin{defn}[Guarded Conway operators]
Let $\BC$ be a guarded co-Cartesian category. We call an operator $(\argument)^\istar$ 
of profile
\begin{align}\label{eq:conway}
f\in\Hom_{\sigma+\id}(X,Y+X)\mapsto f^\dagger\in \Hom_{\sigma}(X,Y)
\end{align} 
a \emph{guarded iteration operator} if it satisfies
\begin{itemize}
  \item\emph{fixpoint:} $f^{\istar} = [\id,f^\istar]\comp f$ for $f:X\to_{\inr} Y+X$;
\end{itemize}
and a \emph{Conway iteration operator} if it additionally satisfies 
\begin{itemize}
  \item\emph{naturality:} $g\comp f^{\istar} = ((g+\id) \comp f)^{\istar}$ for $f:X\to_{\inr} Y+X$, $g : Y \to Z$;
  \item\emph{dinaturality:} $([\inl, h]\comp g)^{\istar} = [\id, ([\inl, g] \comp h)^{\istar}] \comp g$ for  $g : X \to_{\inr} Y + Z$ and $h:Z\to Y+X$ or $g : X \to Y + Z$ and $h:Z\to_{\inr} Y+X$;
  \item\emph{(co)diagonal:} $([\id,\inr] \comp f)^{\istar} = f^{\istar\istar}$ for  $f : X \to_{\inr+\id} (Y + X) + X$.
\end{itemize}
Furthermore, we distinguish the following principles:
\begin{itemize}
  \item\emph{squaring~\cite{Esik99}:} $f^{\istar} = ([\inl,f]\comp f)^\istar$ for $f:X\to_{\inr} Y+X$;
  \item\emph{uniformity w.r.t.\ a subcategory $\BS$ of $\BC$}:  
$(\id+ h)\comp f = g\comp h$ implies $f^{\istar} = g^{\istar}\comp h$ 
for all $f:X\to_{\inr} Z+X$, $g:Y\to_{\inr} Z+Y$ and $h: Y\to X$ from $\BS$;
\end{itemize}
and call $(\argument)^\istar$ \emph{squarable} or \emph{uniform}
if it satisfies squaring or uniformity, respectively.
\end{defn}
\emph{Guarded (Conway) recursion operators} $(\argument)_{\istar}$ on
guarded Cartesian categories are defined dually in a straightforward
manner. %
We collect the following facts about guarded iteration operators for further 
reference.
\begin{lem}\label{lem:iter}
Let $(\argument)^\istar$ be a guarded iteration operator on $(\BC,+,\iobj)$. 
\begin{enumerate}
\item If $(\argument)^\istar$ is uniform w.r.t.\ some co-Cartesian
  subcategory of $\BC$ and satisfies the codiagonal identity then
  it is squarable.
\item If $(\argument)^\istar$ is squarable and uniform w.r.t.\
  coproduct injections then it  is dinatural.
 \item If $(\argument)^\istar$ is Conway then 
it is uniform w.r.t.\  coproduct injections.
\end{enumerate}
\end{lem}
\begin{prop}\label{prop:g-co-cart}
  A guarded co-Cartesian category\/ $\BC$ is traced iff it is equipped
  with a guarded Conway iteration operator $(\argument)^\istar$, with
  mutual conversions like in the total case~\cite{Hasegawa97,Hasegawa99}.
\end{prop}
\begin{expl}[Guarded Conway operators] 
  We list some examples of guarded Conway iteration/recursion
  operators. In all cases except~\ref{item:conway}, Conwayness follows
  from uniqueness of fixpoints~\cite[Theorem
  17]{GoncharovSchroderEtAl17}.
\begin{cenumerate}
\item In a vacuously guarded co-Cartesian category
  (Remark~\ref{rem:triv-cocartesian}), $f:X\to_{\inr} Y+Z$ iff
  $f=\inl g$ for some $g:X\to Y$.  If coproduct injections are monic,
  then $g$ is uniquely determined, and $f^\istar = g$ defines a
  guarded Conway operator.
\item\label{item:conway} Every Cartesian category $\BC$ is guarded
  under $\Hom^{\pi}(X,Y) = \Hom(X,Y)$ (making every morphism
  guarded). Then $\BC$ has a guarded Conway recursion operator iff
  $\BC$ is a \emph{Conway category}~\cite{Esik15}, i.e.\ models
  standard total recursion.
\item The guarded Cartesian category of complete metric spaces as in
  Example~\ref{expl:m-space} is traced: For
  $f:X\times Y\to^{\pr_2} Y$, define $f^\dagger(x)$ as the unique
  fixpoint of $\lambda y.\,f(x,y)$ according to Banach's fixpoint
  theorem.
\item Similarly, the topos of trees, ideally guarded as in
  Example~\ref{expl:tt}, has a guarded Conway recursion operator
  obtained by taking unique
  fixpoints~\cite[Theorem~2.4]{BirkedalMgelbergEtAl12}.
\item The guarded co-Cartesian category $\BC_{\BBT_{\Sigma}}$ of
  side-effecting processes (Example~\ref{expl:pa}) has a guarded
  Conway iteration operator obtained by taking unique fixpoints,
  thanks to the universal property of the final coalgebra
  $T_{\Sigma}X$~\cite{PirogGibbons14}.
\end{cenumerate}
\end{expl}
\noindent\textbf{Guarded vs. unguarded recursion}
We proceed to present a class of examples relating guarded and unguarded
recursion. For motivation, consider the category $(\Cpo,\times,1)$ of
complete partial orders (cpos) and continuous maps.  This category nearly
supports recursion via least fixpoints, except that, e.g.,
$\id:X\to X$ only has a least fixpoint if $X$ has a bottom. The following equivalent approaches involve the \emph{lifting
  monad} $(\argument)_{\bot}$, which adjoins a fresh bottom 
$\bot$ to a given \mbox{$X\in |\Cpo|$}.
\begin{citemize}
\item[]\hspace{-1ex}\emph{Classical
    approach}~\cite{Winskel93,SimpsonPlotkin00}: Define a total
  recursion operator $(-)_\iistar$ on the category $\Cpo_{\bot}$ of
  \emph{pointed cpos} and continuous maps, using least fixpoints.
\item[]\hspace{-1ex}\emph{Guarded approach} (cf.\,\cite{MiliusLitak17}):
  Extend $\Cpo$ to a guarded category: $f:X\times Y\to^{\snd} Z$ iff
  $f\in \{g\comp (\id\times\eta)\mid g:X\times Y_{\bot}\to Z\}$ (see
  Proposition~\ref{prop:mult-guard}), and define a guarded recursion
  operator sending $f=g\comp (\id\times\eta):Y\times X\to^{\snd} X$ to
  $f_{\istar} = g\comp\brks{\id,\hat f} :Y\to X$ with
  $\hat f(y)\in X_{\bot}$ calculated as the least fixpoint of
  $\lambda z.\, \eta g(y, z)$.
\end{citemize}
Pointed cpos happen to be always of the form $X_{\bot}$ with
$X\in |\Cpo|$, which indicates that $(\argument)_{\iistar}$ is a
special case of $(\argument)_{\istar}$. This is no longer true in more
general cases when the connection between $(\argument)_{\iistar}$ and
$(\argument)_{\istar}$ is more intricate. We show that
$(\argument)_{\iistar}$ and $(\argument)_{\istar}$ are nevertheless
equivalent under reasonable assumptions.
\begin{defn}[\cite{CrolePitts90}] \label{def:let-ccc}
A \emph{let-ccc with a fixpoint object} is a tuple $(\BC, \BBT, \Omega, \omega)$,
consisting of a Cartesian closed category $\BC$, a strong monad $\BBT$ on it, 
an initial $T$-algebra  $(\Omega,\ini)$ and an equalizer $\omega:1\to\Omega$ 
of $\ini\eta:\Omega\to\Omega$ and $\id:\Omega\to\Omega$.
\end{defn}
The key requirement is the last one, satisfied, e.g., for $\Cpo$ and
the lifting monad.
Given a monad $\BBT$ on $\BC$, $\BC^{\BBT}_{\star}$ denotes the
category of $\BBT$-algebras and $\BC$-morphisms (instead of
$\BBT$-algebra homomorphisms).
\begin{prop}[{\cite[Theorem~4.6]{Simpson92}}]\label{prop:simpson}
Let $(\BC, \BBT, \Omega, \omega)$ be a let-ccc with a fixpoint object. 
Then $\BC^{\BBT}_{\star}$ has a unique $\BC^{\BBT}$-uniform recursion operator $(\argument)_{\iistar}$. 
\end{prop}
By~\cite[Theorem 4]{SimpsonPlotkin00}, the operator~$(\argument)_{\iistar}$
in Proposition~\ref{prop:simpson} is Conway, in particular, by Lemma~\ref{lem:iter}, squarable, 
if $\BC$ has a natural numbers object and $\BBT$ is an \emph{equational lifting
monad}~\cite{BucaloFuhrmannEtAl03}, such as $(-)_\bot$. %
There are however further squarable operators obtained via
Proposition~\ref{prop:simpson}, e.g.\ for the partial state monad
$TX = {(X\times S)^S_{\bot}}$~\cite{CrolePitts90}. By
Lemma~\ref{lem:iter}, the following result applies in particular in
the setup of Proposition~\ref{prop:simpson} under the additional
assumption of squarability.
\begin{thm}\label{thm:grec_from_rec}
  Let $\BBT$
  be a strong monad on a Cartesian category~$\BC$.
  The following gives a bijective correspondence between squarable
  dinatural recursive operators $(\argument)_{\iistar}$
  on $\BC^\BBT_\star$
  and squarable dinatural guarded recursive operators
  $(\argument)_{\istar}$
  on $\BC$
  ideally guarded over $\IHom(X,Y) = \{f\comp\eta\mid f:TX\to Y\}$:
\begin{align}
\label{eq:def-iistar}(f:B\times A\to A)_{\iistar} =\;&\algebra\comp(\eta f(\id\times\algebra))_{\istar} && \text{for $(A,\algebra)\in|\BC^\BBT_\star|$}\\
\label{eq:def-istar}(f=g\comp (\id\times\eta):Y\times X\to X)_{\istar} =\;& g\brks{\id,(\eta g)_{\iistar}}
\end{align}
(in~\eqref{eq:def-istar} we call on a slight extension of
$(\argument)_{\iistar}$
(Lemma~\ref{lem:par_ext}); the right hand side
of~\eqref{eq:def-iistar} is defined because $\eta
f(\id\times\algebra)$ factors as $\eta f(\id\times\algebra
(T\algebra)\eta)$).  Moreover,
$(\argument)_\istar$ is Conway iff so is $(\argument)_\iistar$.
\end{thm}

\section{Vacuous Guardedness and Nuclear Ideals}\label{sec:dag}
\noindent We proceed to discuss traces in vacuously guarded categories
(Lemma~\ref{lem:triv}), and show that the partial trace operation in
the category of (possibly infinite-dimensional) Hilbert
spaces~\cite{AbramskyBluteEtAl99} in fact lives over the vacuous
guarded structure. We first note that vacuous guarded structures are
traced as soon as a simple rewiring operation satisfies a suitable
well-definedness condition (similar to one defining traced nuclear
ideals~\cite[Definition~8.14]{AbramskyBluteEtAl99}):

\begin{prop}\label{prop:trivial-trace}
  Let $(\BC,\tensor,I)$ be vacuously guarded. If for
  $f\in\GHom(A\tensor B\comma{C\tensor D})$ with factorization
  $f=(h\tensor\id_{D\tensor U})(\id_{A\tensor U}\tensor g)$ (eliding
  associativity), $g:B\to E\tensor D\tensor U$,
  $h:A \tensor U\tensor E\to C$ as per Lemma~\ref{lem:triv}, the
  composite
  \begin{equation}\label{eq:trivial-trace}
    A\tensor B\xto{\id_A\tensor g} A\tensor E\tensor D\tensor U \cong A\tensor U\tensor E\tensor D\xto{h\tensor\id_D} C\tensor D
  \end{equation}
  depends only on $f$, then $\BC$ is guarded traced, with
  $\tr^U_{A,B,C,D}(f)$ defined as~\eqref{eq:trivial-trace}.
\end{prop}
Diagrammatically, the trace in a vacuously guarded category is thus given
by
\begin{center}
  \includegraphics[scale=.3]{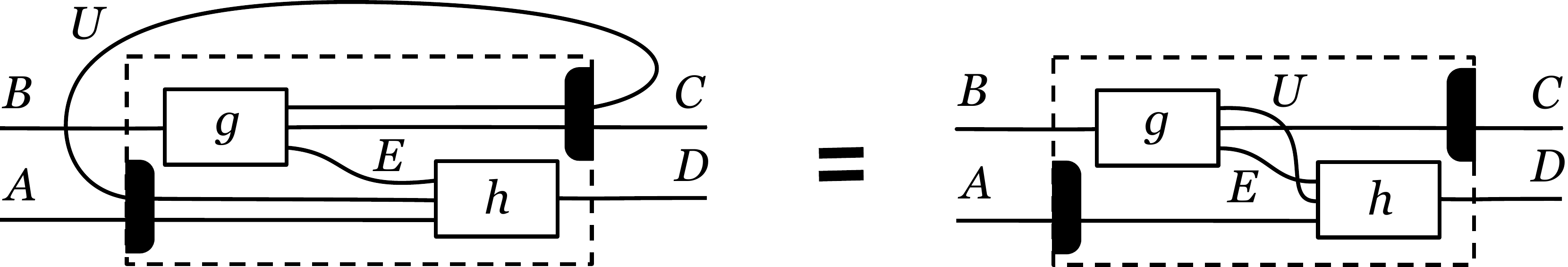}
\end{center}
We proceed to instantiate the above to Hilbert spaces. On a more
abstract level, a \emph{dagger symmetric monoidal
  category}~\cite{Selinger07} (or \emph{tensored
  $*$-category}~\cite{AbramskyBluteEtAl99}) is a symmetric monoidal
category~$(\BC,\tensor,I)$ equipped with an identity-on-objects
strictly involutive functor $(\argument)^\dagger:\BC\to\BC^{op}$
coherently preserving the symmetric monoidal structure. %
The main motivation for
dagger symmetric monoidal categories is to capture categories that are
similar to (dagger) compact closed categories in that they admit a
canonical trace construction for certain morphisms, but fail to be
closed, much less compact closed. The ``compact closed part'' of a
dagger symmetric monoidal category is axiomatized as follows.
\begin{defn}[Nuclear Ideal,~\cite{AbramskyBluteEtAl99}]\label{def:nucl}
  A \emph{nuclear ideal} $\N$ in a dagger symmetric monoidal category
  $(\BC,\tensor, I, (\argument)^\dagger)$ is a family of subsets
  $\N(X,Y)\subseteq\Hom_{\BC}(X,Y)$, $X,Y\in |\BC|$, satisfying the
  following conditions:
\begin{cenumerate}
\item $\N$ is closed under $\tensor$, $(\argument)^\dagger$, and
  composition with arbitrary morphisms on both sides;
\item There is a bijection
  $\theta:\N(X,Y)\to\Hom_{\BC}(I,X^\dagger\tensor Y)$, natural in~$X$
  and~$Y$, coherently preserving the dagger symmetric monoidal
  structure. %
\item (\emph{Compactness}) For $f\in\N(B,A)$ and $g\in\N(B,C)$, the
  following diagram commutes:
\begin{equation*}
\begin{tikzcd}
A
	\ar[r, "\cong"]
	\ar[d, "g\comp f^\dagger"']& 
A\tensor I
	\ar[r, "\id_A\tensor\theta(g)"] &[3em]
A\tensor (B^\dagger\tensor C)
	\ar[d, "\cong"] \\[-.5em]
C&
I\tensor C
	\ar[l, "\cong"']&
( B^\dagger\tensor A)\tensor C
	\ar[l, "(\theta(f))^\dagger\tensor\id_C"]
\end{tikzcd}
\end{equation*}
\end{cenumerate} 
\end{defn}
The above definition is slightly simplified in that we elide a
covariant involutive functor $\overline{(\argument)}:\BC\to\BC$,
capturing, e.g.\ complex conjugation; i.e., we essentially restrict to
spaces over the reals.

We proceed to present a representative example of a nuclear ideal in
the category of Hilbert
spaces. %
Recall that a \emph{Hilbert space}~\cite{KadisonRingrose97} $H$ over
the field $\real$ of reals is a vector space with an \emph{inner
  product} $\brks{\argument,\argument}:H\times H\to \real$ that is
complete as a \emph{normed space} under the induced \emph{norm}
$\norm{x}=\sqrt{\brks{x,x}}$.  %
Let $\Hilb$ be the category of Hilbert spaces and bounded linear
operators. %

Clearly, $\real$ itself is a Hilbert space; linear operators $X\to\real$
are conventionally called \emph{functionals}. More generally, we
consider \emph{(multi-)linear} functionals
$X_1\times \ldots\times X_n\to \real$, i.e.\ maps that are linear in every
argument.  Such a functional is \emph{bounded} if
$|f(x_1,\ldots,x_n)|\leq c\comp \norm{x_1}\cdots \norm{x_n}$ for some
constant $c\in\real$. %
We can move between bounded linear operators and bounded linear
functionals, similarly as we can move between relations and functions
to the Booleans:
\begin{prop}[{\cite[Theorem 2.4.1]{KadisonRingrose97}}]\label{prop:fun_op}
  Given a bounded linear operator $f:X\to Y$,
  $f^\circ(x,y) = \brks{fx, y}$ defines a bounded linear
  functional~$f^\circ$, and every bounded linear functional
  $X\times Y\to\real$ arises in this way.
\end{prop}
\begin{defn}[Hilbert-Schmidt operators/functionals]
  A bounded linear functional $f:X_1\times \ldots\times X_n\to \real$ is
  \emph{Hilbert-Schmidt} if the sum
  \begin{displaymath}\textstyle
    \sum\nolimits_{x_1\in B_1}\ldots\sum\nolimits_{x_n\in B_n} (f(x_1, \ldots, x_n))^2
  \end{displaymath}
  is finite for some, and then any, orthonormal bases $B_1,\ldots,B_n$
  of $X_1,\ldots,X_n$, respectively. A bounded linear operator
  $f:X\to Y$ is \emph{Hilbert-Schmidt} if the induced functional
  $f^\circ$ (Proposition~\ref{prop:fun_op}) is Hilbert-Schmidt,
  equivalently if $\sum_{x\in B} \norm{f x}^2$ is finite for some, and
  then any, orthonormal basis $B$ of $X$. We denote by $\HS(X,Y)$ the
  space of all Hilbert-Schmidt operators from $X$ to $Y$.
\end{defn}
For $X,Y\in|\Hilb|$, the space of Hilbert-Schmidt functionals
$X\times Y\to\real$ is itself a Hilbert space, denoted $X\tensor Y$,
with the pointwise vector space structure and the inner product
$\brks{f,g} = \sum\nolimits_{x\in B}\sum\nolimits_{y\in B'} f(x,
y)\comp g(x, y)$. %
where $B$ and $B'$ are orthonormal bases of $X$ and $Y$, respectively.
By virtue of the equivalence between $f$ and $f^\circ$, this induces a
Hilbert space structure on $\HS(X,Y)$, with induced norm
$\norm{f}_2 =\sqrt{\sum\nolimits_{x\in B} \norm{fx}^2}$. The operator
$\tensor$ forms part of a dagger symmetric monoidal structure on
$\Hilb$, with unit $\real$. For a bounded linear operator $f:X\to Y$,
$f^\dagger:Y\to X$ is the \emph{adjoint operator} uniquely determined
by equation $\brks{x, f^\dagger y} = \brks{fx, y}$.  The tensor
product of $f:A\to B$ and $g:C\to D$ is the functional sending
$h:A\times C\to\real$ to
$h\comp (f^\dagger\times g^\dagger):B\times D\to\real$.  Given
$a\in A$ and $c\in C$, let us denote by $a\tensor c\in A\tensor C$ the
functional $(a',c')\mto \brks{a,a'}\comp\brks{c,c'}$, and so, with the
above $f$ and $g$, $(f\tensor g)(a\tensor c) = f(a)\tensor g(c)$.
\begin{prop}\label{prop:hs-nucl}
  \cite{AbramskyBluteEtAl99} The Hilbert-Schmidt operators form a
  nuclear ideal in $\Hilb$ with
  $\theta:\HS(X,Y)\cong\Hom(\real,X^\dagger\tensor Y)$ defined by
\begin{equation*}
\theta(f:X\to Y)(r:\real)(x:X, y:Y) = r\comp\brks{fx,y}.
\end{equation*}
\end{prop}
A crucial fact underlying the proof of Proposition~\ref{prop:hs-nucl}
is that $\HS(X,Y)$ is isomorphic to $X^\dagger\tensor Y$, naturally in
$X$ and $Y$. We emphasize that what makes the case of $\Hilb$
significant is that we do not restrict to finite-dimensional Hilbert
spaces. In that case all bounded linear operators would be
Hilbert-Schmidt and the corresponding category would be (dagger)
compact closed~\cite{Selinger07}. In the infinite-dimensional case,
identities need not be Hilbert-Schmidt, so $\HS$ is indeed only an
ideal and not a subcategory.

Let $\N^2(X,Y) = \{g^\dagger h:X\to Y\mid h\in\N(X,Z), g\in\N(Y,Z)\}$ for any 
nuclear ideal $\N$. The main theorem of the section now can be stated as follows.
\begin{thm}\label{thm:hilbert-schmidt}
\begin{cenumerate}
\item\label{item:hs2} The guarded ideal induced by the vacuous guarded
  structure on $\Hilb$ (see~\eqref{eq:trivial-ideal}) is precisely $\HS^2$, and $\Hilb$ is guarded
  traced over $\HS^2$.
\item\label{item:hilb-trace-dagger} Guarded traces in $\Hilb$ commute
  with $(-)^\dagger$ in the sense that if
  $f\in\GHom((A\tensor U)\tensor B,C\tensor (D\tensor U))$, then
  $\gamma_{B,A\tensor U} f^\dagger \gamma_{D\tensor U,C}\in
  \GHom((D\tensor U)\tensor C,B\tensor (A\tensor U))$
  and
  $\tr^U_{D,C,B,A}(\gamma_{B,A\tensor U} f^\dagger \gamma_{D\tensor
    U,C}) =
  \gamma_{A,B}\comp(\tr_{A,B,C,D}^U(f))^\dagger\comp\gamma_{C,D}$.
\end{cenumerate}
\end{thm}
Clause~\ref{item:hs2} is a generalization of the result
in~\cite[Theorem 8.16]{AbramskyBluteEtAl99} to parametrized
traces. Specifically, we obtain agreement with the conventional
mathematical definition of trace:
given $f\in\HS^2(X,X)$,  $\tr(f) = \sum_i \brks{f(e_i), e_i}$
for any choice of an orthonormal basis $(e_i)_i$, and $\HS^2(X,X)$
contains precisely those $f$ for which this sum is absolutely
convergent independently of the basis. 

\section{Conclusions and Further Work}

\noindent We have presented and investigated a notion of abstract
\emph{guardedness} and guarded \emph{traces}, focusing on foundational
results and important classes of examples. We have distinguished a
more specific notion of \emph{ideal guardedness}, which in many
respects appears to be better behaved than the unrestricted one, in
particular ensures closer agreement between structural and geometric
guardedness. An unexpectedly prominent role is played by `vacuous'
guardedness, characterized by the absence of paths connecting
unguarded inputs to guarded outputs; e.g., partial traces in Hilbert
spaces~\cite{AbramskyBluteEtAl99} turn out to be based on this form of
guardedness. Further research will concern a coherence theorem for
guarded traced categories generalizing the well-known unguarded
case~\cite{JoyalStreetEtAl96,Selinger04}, and a generalization of the
Int-construction~\cite{JoyalStreetEtAl96}, which would relate guarded
traced categories to a suitable guarded version of compact closed
categories. Also, we plan to investigate guarded traced categories as
a basis for generalized Hoare logics, extending and unifying previous
work~\cite{ArthanMartinEtAl09,GoncharovSchroder13a}.

\newpage
\bibliographystyle{myabbrv}
\bibliography{monads}
\newpage
\appendix

\section{Appendix: Omitted Details and Proofs}

\subsection{Derivability of Weakening (Section~\ref{sec:guarded})}
\noindent We show that we can weaken on the right (output) side; by
duality, we can then also weaken on the input side, and the claim
follows by weakening first on the output and then on the input
side. That is, we assume that
$f\in\GHom(A\tensor B, C\tensor (D'\tensor D))$ and derive
$f\in\GHom(A\tensor B,(C\tensor D')\tensor D)$.

First note that by \textbf{(cmp${}_{\tensor}$)} and
\textbf{(vac${}_{\tensor}$)}, guardedness annotations are stable under
rearranging guarded output gates via monoidal isomorphims, and
similarly for the unguarded output gates and both types of input
gates. We obtain by \textbf{(vac${}_{\tensor}$)} that
$\id_C\tensor I\in\GHom(C\tensor I,C\tensor I)$ and
$I\tensor\id_D\in\GHom(I\tensor D,I\tensor D)$. By
\textbf{(uni${}_{\tensor}$)}, \textbf{(par${}_{\tensor}$)}, and
stability under monoidal isomorphisms, we derive
\begin{equation*}
  \id_C\tensor \gamma_{I,D'}\tensor  \id_D\in\GHom((C\tensor I)\tensor (D'\tensor D),(C\tensor D')\tensor(I\tensor D)),  
\end{equation*}
(eliding associativity throughout) and hence, again using stability
under monoidal isomorphisms,
\begin{equation*}
  \id_C\tensor \id_{D'}\tensor  \id_D\in\GHom(C \tensor (D'\tensor D),(C\tensor D')\tensor D).  
\end{equation*}
Our goal then follows by \textbf{(cmp${}_{\tensor}$)}.

\subsection{Proof of Theorem~\ref{thm:gcompl}}
For purposes of this proof, call a path leading from an input gate
in~$A$ to an output gate in $D$ as in the claim \emph{critical}. That
is, we are to show that $e$ types as requested iff all critical paths
in its diagram are guarded.

\emph{`Only if':} By induction on the derivation of
$e\in\GHom(A\tensor B,C\tensor D)$. The base case (introduction of
morphism symbols) is trivial. The cases for the rules from
Definition~\ref{def:guard_sm}, diagrammatically represented according
to Figure~\ref{fig:gmon}, are as follows. In the cases for rules
\textbf{(uni${}_{\tensor}$)} and \textbf{(vac${}_{\tensor}$)}, there
are no critical paths. For rule \textbf{(par${}_{\tensor}$)}, just
note that every critical path in the diagram for $f\tensor g$ is
either a critical path in the diagram for $f$ or a critical path in
the diagram for~$g$. For \textbf{(cmp${}_{\tensor}$)}, let $\pi$ be a
critical path in the diagram for $gf$. We distinguish cases on whether
$\pi$ leaves $f$ through a guarded or an unguarded output gate. By the
symmetry manifest in Figure~\ref{fig:gmon}, we can w.l.o.g.\ assume
the latter. As can, again, be seen in Figure~\ref{fig:gmon}, $\pi$
then enters $g$ through an unguarded input gate and leaves~$g$ through
a guarded output, so by the inductive hypothesis, the part of~$\pi$
that leads through $g$ is guarded, and then of course $\pi$ itself is
guarded.

\emph{`If':} We can regard the diagrammatic rules in
Figure~\ref{fig:gmon} as a set of rules for establishing guardedness
of diagrams (essentially, this lets us use the known coherence theorem
for symmetric monoidal categories to avoid bookkeeping with
associativity etc.). In terms of diagrams, object expressions (such as
$A$ and $D$ in the claim) correspond to sets of gates, and we will
henceforth conflate the two notions. Let us denote by $G_e(E,O)$ the
statement that the diagram of $e$ is provably unguarded in a set $E$
of input gates and simultaneously guarded in a set $O$ of output gates
(i.e.\ that the corresponding gates can be marked black according to
the rules in Figure~\ref{fig:gmon}). We thus have to show
$G_{e}(A,D)$.  We proceed by structural induction over $e$. For the
case where $e$ is a basic box~$f$, note that the assumption implies
that the unguarded gates of the given diagram are contained in those
given in the basic guardedness assumption for $f$, similarly for the
guarded outputs, so that $G_e(A,D)$ by weakening. The other base cases
are straightforward, as they do not contain any basic boxes, so that
the assumption implies that there are no critical paths; to make one
example implicit: if $e$ is an identity, then absence of critical
paths implies that one of $A$ and $D$ is empty, so that $G_e(A,D)$ by
\textbf{(vac${}_{\tensor}$)}. The other cases are as follows.
\begin{citemize}
\item The expression $e$ is a composite $e_2\comp e_1$. Let $M$ be the
  set of joint gates $W$ of~$e_1$ and $e_2$ such that all paths from
  $W$ to gates in $D$ in the diagram of $e_2$ are guarded, and
  analogously, let $N$ be the set of joint gates $W$ of $e_1$ and
  $e_2$ such that all paths from gates in $A$ to $W$ are guarded. Note
  that the union $N\cup M$ consists of all joint gates of $e_1$ and
  $e_2$: If there was a joint gate $W\notin N\cup M$, then there would
  be an unguarded path from some $X$ in $A$ to $W$ and an unguarded
  path from $W$ to some $Y$ in $D$; then the concatenated path would
  be critical (for $e_2e_1$) and also unguarded, contradicting the
  assumption. Now by induction $G_{e_1}(A,N)$ and $G_{e_2}(M,D)$, and
  by the above, the complement $\overline N$ of $N$ is
  $\overline N= (M\cup N)\setminus N\subseteq M$. By weakening, we
  therefore have $G_{e_2}(\overline N,D)$, so $G_e(A,D)$ by (the
  diagrammatic version of)~\textbf{(cmp${}_{\tensor}$)}.
\item The expression $e$ is a tensor $e_1\tensor e_2$. Then $A$ and
  $D$ are disjoint unions $A=A_1\cup A_2$, $D=D_1\cup D_2$ where the
  gates in $A_1$ and $D_1$ are contributed by $e_1$ and those in $A_2$
  and $D_2$ by $e_2$. Every path in the diagram of $e_1$ from a gate
  in $A_1$ to a gate in $D_1$ is a critical path in $e_1\tensor e_2$,
  hence guarded by assumption; hence $G_{e_1}(A_1,D_1)$ by
  induction. Analogously, $G_{e_2}(A_2,D_2)$, and thus $G_e(A,D)$ by
  (the diagrammatic version of)~\textbf{(par${}_{\tensor}$)}.  \qed
\end{citemize}

\subsection{Proof of Lemma~\ref{lem:triv}}

Any morphism that factors as $(h\tensor\id_D)(\id_A\tensor g)$ as in
the statement is guarded in any guarded structure by
rules~\textbf{(vac${}_{\tensor}$)} and~\textbf{(cmp${}_{\tensor}$)}
(plus weakening). This proves that the putative guarded structure
described is contained in all guarded structures on $\BC$. It remains
to show that the axioms of Definition~\ref{def:guard_sm} are
satisfied. The rules~\textbf{(uni${}_{\tensor}$)}
and~\textbf{(vac${}_{\tensor}$)}) are clear, and closure under
rule~\textbf{(par${}_{\tensor}$)} is easily seen by rearranging boxes
and gates using commutativity and associativity of~$\tensor$. For
closure under rule~\textbf{(cmp${}_{\tensor}$)}, finally, assume that
$f_1=(h_1\tensor\id)(\id\tensor g_1)$,
$f_2=(h_2\tensor\id)(\id\tensor g_2)$ such that $f_2f_1$ is
defined. Then $f_2f_1$ factors, omitting associativity isomorphisms,
into $\id\tensor((\id\tensor g_2)g_1)$ and
$(h_2(h_1\tensor\id))\tensor\id$.

\subsection{Proof of Proposition~\ref{prop:guard_equiv}}
For brevity, given coproduct injections
$\sigma:X_1\to X_1+X_2\cong X$, $\theta:Y_2\to Y_1+Y_2\cong Y$, we
write 
\begin{equation*}
  f:X\to_{\sigma,\theta} Y\quad\text{for}\quad f\inl:X_1\to_\theta Y,
\end{equation*}
i.e.\
for $f\inl\in\GHom(X_1+X_2,Y_1+Y_2)$.

We need the following lemma:
\begin{lem}\label{lem:guard_pair}
Given $f:X\to_{\sigma,\theta} Z$, $g:Y\to_{\sigma',\theta} Z$, then $[f,g]:X+Y\to_{\sigma+\sigma',\theta} Z$.
\end{lem}

\begin{proof}
  By~\textbf{(par${}_{\tensor}$)},
  $f+g:X+Y\to_{\sigma+\sigma',\theta+\theta} Z+Z$.  Assume w.l.o.g.\
  that $Z=Z'+Z''$, $\theta=\inr:Z''\cpto Z$ and
  $\bar\theta=\inr:Z'\cpto Z$.  Then $Z+Z$ is a coproduct of $Z'+Z'$
  and $Z''+Z''$, hence we obtain
  $f+g:X+Y\to_{\sigma+\sigma',\inr}
  (Z'+Z')+(Z''+Z'')$.
  By~\textbf{(cmp${}_{\tensor}$)},
  $\nabla+\nabla:(Z'+Z')+(Z''+Z'')\to_{\inr,\inr} Z'+Z''$, and so
  by~\textbf{(cmp${}_{\tensor}$)},
  $[f,g]:X+Y\to_{\sigma+\sigma',\theta} Z$.  \qed
\end{proof}

The proof of Proposition~\ref{prop:guard_equiv} then proceeds as
follows. Suppose that $(\BC, +,\iobj)$ is guarded and let us show
first of all that the above condition uniquely determines the
$\Hom_{\inr}(X,Y+Z)$. Indeed, on the one hand we obtain as the
definition: $f\in\Hom_{\inr}(X,Y+Z)$ if
$[f,\bang]\in\GHom(X+\iobj,Y+Z)$.
On the other hand if $f=g\comp\inl$ then $g\in\GHom(X+X',Y+Z)$ implies 
$[f,\bang]=[g\inl,\bang]=[g\inl,\bang\inr]=g\comp (\id+\bang)\in\GHom(X+\iobj,Y+Z)$
by~\textbf{(vac${}_{\tensor}$)} and~\textbf{(cmp${}_{\tensor}$)}, that is,
decomposition of $X$ different than $X+\iobj$ do not affect the definition of $\Hom_{\inr}(X,Y+Z)$.

We proceed to prove the required properties.
\begin{citemize}
  \item \textbf{(vac${}_{\mplus}$)} Let $f:X\to Z$. Then by~\textbf{(vac${}_{\tensor}$)}, 
$f+\bang:X+\iobj\to_{\inr,\inr} Z+Y$. Modulo the fact that $X$ is a coproduct of
$X$ and $\iobj$, this is equivalent to $f\comp\inl:X\to_{\bang,\inr} Z+Y$.
  \item \textbf{(cmp${}_{\mplus}$)} Suppose that $f\in\Hom_{\inr}(X,Y+Z)$, $g\in\Hom_{\sigma}(Y,V)$ and $h:Z\to V$. Then 
$f:X\to_{\bang,\inr} Y+Z$ and we would be done by Lemma~\ref{lem:guard_pair} if we showed
$g:Y\to_{\bang,\sigma} V$ and $h:Z\to_{\id,\sigma} V$. The former of these two
judgements is an assumption. To prove the latter one, note that $h:Z\to_{\id,\id} V$
by~\textbf{(vac${}_{\tensor}$)} and $\id:V\to_{\id,\sigma} V$ by~\textbf{(uni${}_{\tensor}$)}.
Hence, indeed  $h:Z\to_{\id,\sigma} V$ by~\textbf{(cmp${}_{\tensor}$)}.
  \item \textbf{(par${}_{\mplus}$)} The assumption read as $f:X\to_{\bang,\sigma} Z$, $g:Y\to_{\bang,\sigma} Z$.
By Lemma~\ref{lem:guard_pair}, $[f,g]:X+Y \to_{\bang,\sigma} Z$, which is the goal.
\end{citemize}

We proceed to show the converse implication.
\begin{citemize}
 \item \textbf{(uni${}_{\tensor}$)} Note that $\gamma_{\iobj,A}=[\bang,\inl]:\iobj + A\to A+\iobj$,
hence $\gamma_{\iobj,A}\inl = \bang = \inl\bang\in\Hom_{\inr}(A+\iobj)$ 
by~\textbf{(vac${}_{\mplus}$)}.
\item\textbf{(vac${}_{\tensor}$)} Let $f:A\to B$ and $g:C\to D$. Note that $\inl:A\to_{\inr} A+C$ and $\inl\comp f:A\to_{\inr} B+D$
by~\textbf{(vac${}_{\mplus}$)}, and therefore, by~\textbf{(cmp${}_{\mplus}$)}, $(f+g)\inl = 
[\inl f,\inr g]\inl:A\to_{\inr} B+D$, which is equivalent to the goal.
 \item\textbf{(cmp${}_{\tensor}$)} Since
$f\comp g\inl = [f\inl, f\inr]\comp g\inl$ and by assumption
$g\inl\in\Hom_{\inr}(A, C+D)$, by~\textbf{(cmp${}_{\mplus}$)},
we reduce the goal to $f\inl\in\Hom_{\inr}(C, E+F)$, which is again part of the
assumption.
 \item\textbf{(par${}_{\tensor}$)} By assumption, $f\inl\in\Hom_{\inr}(A, C+D)$ 
and\/ $g\inl\in\Hom_{\inr}(A', C'+D')$. And we need to show that for 
$[(\inl+\inl)\comp f,(\inr+\inr)g](\inl+\inl)\in\Hom_{\inr}(A+A', (C+C')+(D+D'))$.	
Indeed, by assumption $f\inl\in\Hom_{\inr}(A, C+D)$ and hence by~\textbf{(vac${}_{\mplus}$)} 
and~\textbf{(cmp${}_{\mplus}$)}, $(\inl+\inl)\comp f\inl=[\inl\inl,\inr\inl]\comp f\comp\inl\in\Hom_{\inr}(A, (C+C')+(D+D'))$.
Symmetrically, $(\inr+\inr)\comp g\inl\in\Hom_{\inr}(A', (C+C')+(D+D'))$,
and thus we are done by~\textbf{(par${}_{\mplus}$)}.\qed
\end{citemize}

\subsection{Proof of Lemma~\ref{lem:total}}
Note that by~\textbf{(vac${}_{\tensor}$)}, $\IHom(I,I)=\Hom(I,I)$
The closure conditions are  
instances of diagrams from Figure~\ref{fig:gmon}.

\begin{center}
\includegraphics[width=.20\textwidth]{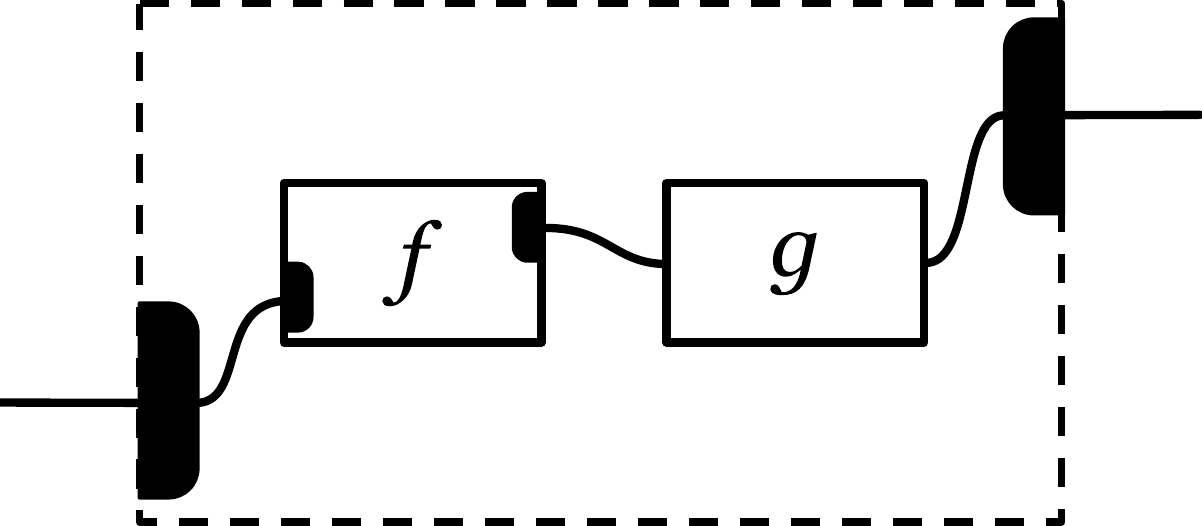}
\qquad\qquad
\includegraphics[width=.20\textwidth]{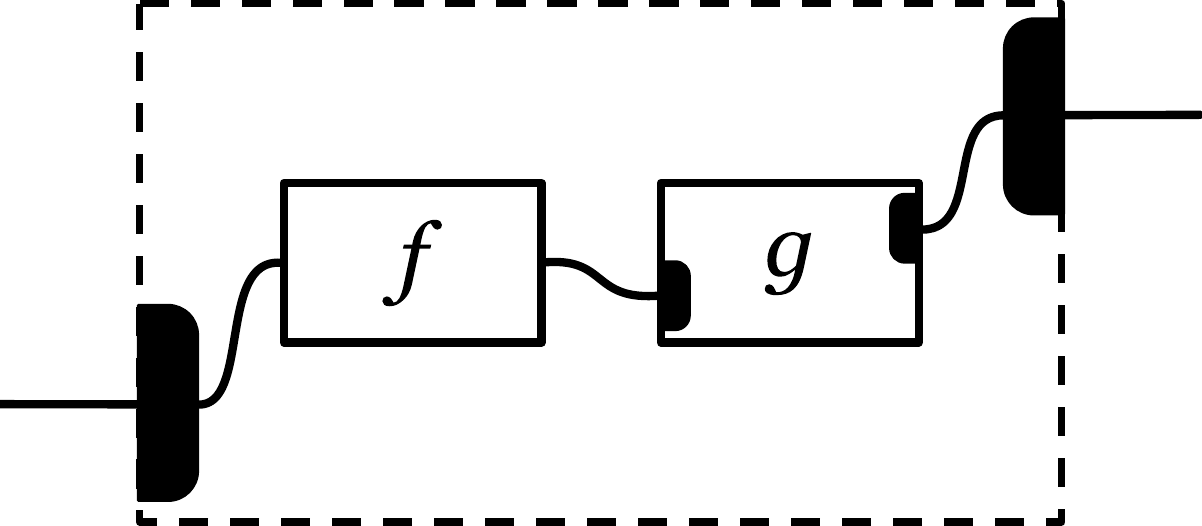}
\qquad\qquad
\includegraphics[width=.20\textwidth]{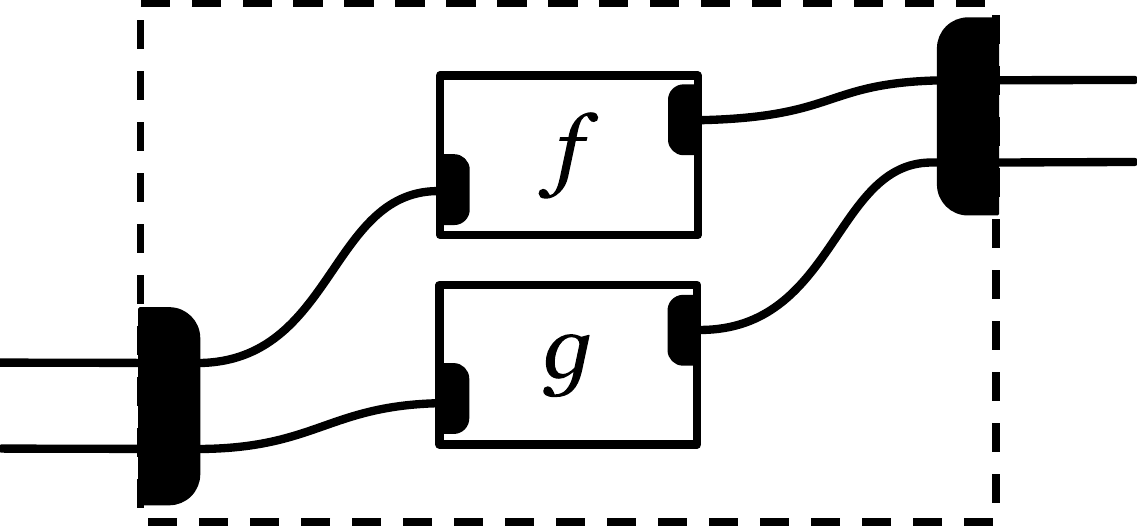}
\end{center}
\qed

\subsection{Existence of Non-Ideal Guarded Structures
  (Section~\ref{sec:ideals})}

\begin{expl}\label{expl:counter}
  Let $\BBT$ be the monad on $\Set$ for the algebraic theory of
  commutative semigroups with the additional law
$x * y = x$.
The Kleisli category $\Set_{\BBT}$ is co-Cartesian with coproducts
inherited from $\Set$, and so we put $f:X\to_2 T(Y+Z)$ iff
$f = (T\inl)\comp g$ for some $g:X\to TY$. According to this
definition, $f\in\IHom(X,TY)$ iff $f$ factors through
$T\emptyset =\emptyset$, i.e.\ when $X=\emptyset$ and $f=\bang_{TY}$.
This induces a different guarded category structure on $\Set_{\BBT}$:
$f:X\to_2 T(Y+Z)$ iff $Z=\emptyset$. Consider the term $x*y\in T(X+Y)$
(seen as a morphism $1\to T(X+Y)$) with $x\in X$ and $y\in Y$. It is
$\inj_2$-guarded under the original definition, for it is equivalent
to the term $x\in T(X+\emptyset)$, but not under the new definition
unless $Y=\emptyset$.
\end{expl}

\subsection{Proof of Theorem~\ref{thm:t_to_p}}
By the axioms of guarded categories (or more quickly by
Theorem~\ref{thm:gcompl}), it is clear that morphisms of the
given form must be in $\GHom(A\tensor B,C\tensor D)$. It remains to
check closure under the axioms of Definition~\ref{def:guard_sm}. To that
end, consider a generic morphism $(q\tensor\id)\comp w_n\comp\ldots w_1\comp (\id\tensor p)$ where
$p:B\to B_1\tensor C_1$, $q:B_n\tensor C_n\to C$, $A_1=A$, $C_n=D$
and each $w_i:A_i\tensor (B_i\tensor C_i)\to A_{i+1}\tensor (B_{i+1}\tensor
C_{i+1})$ has the form
\begin{equation}\label{eq:tguard_i}
\vcenter{\hbox{\includegraphics[scale=.3]{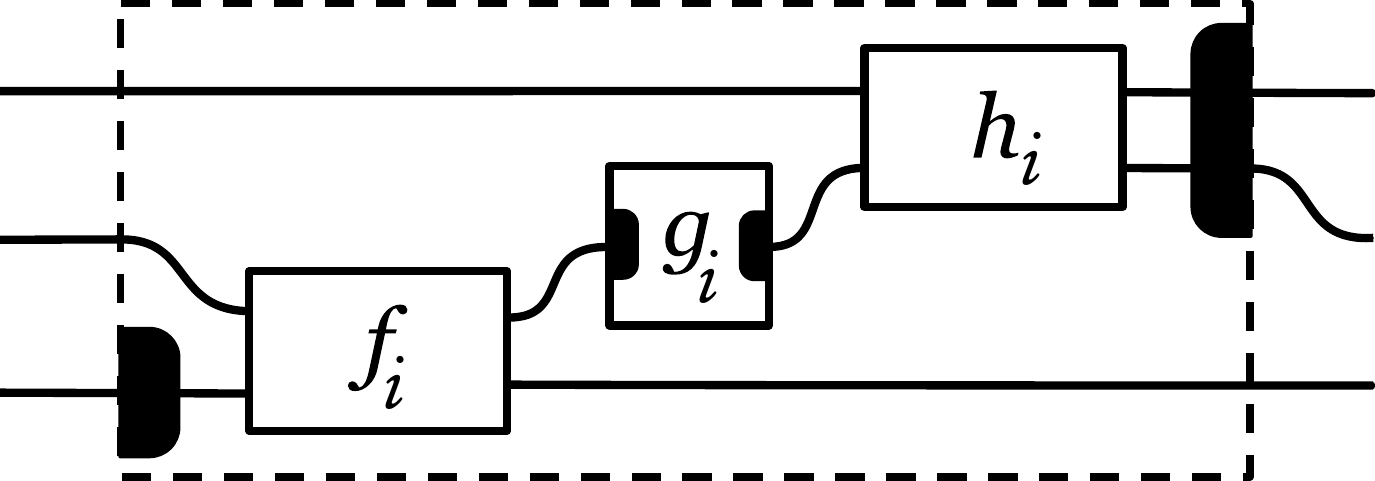}}}
\medskip
\end{equation}

In
order to capture~\textbf{(uni${}_{\tensor}$)}
and~\textbf{(vac${}_{\tensor}$)} it suffices to take $n=0$, and select
$p$ and $q$ in the obvious way.  Axiom~\textbf{(cmp${}_{\tensor}$)} is
clear by definition. Let us verify~\textbf{(par${}_{\tensor}$)}. Given
$f=(q\tensor\id)\comp w_n\comp\ldots w_1\comp (\id\tensor p)$ and
$f'=(q'\tensor\id)\comp w_m'\comp\ldots w_1'\comp (\id\tensor p')$, we
assume them to be an input to the~\textbf{(par${}_{\tensor}$)}
rule. W.l.o.g.\ we assume that $n=m$ (missing sections of the
form~\eqref{eq:tguard_i} with a middle wire of type $I$ can obviously
be added by need either to $f$ or to $f'$). Note that the tensor
product of two sections of the form~\eqref{eq:tguard_i} can again be
arranged in a diagram in the same form:
\medskip
\begin{displaymath}
\includegraphics[scale=.3]{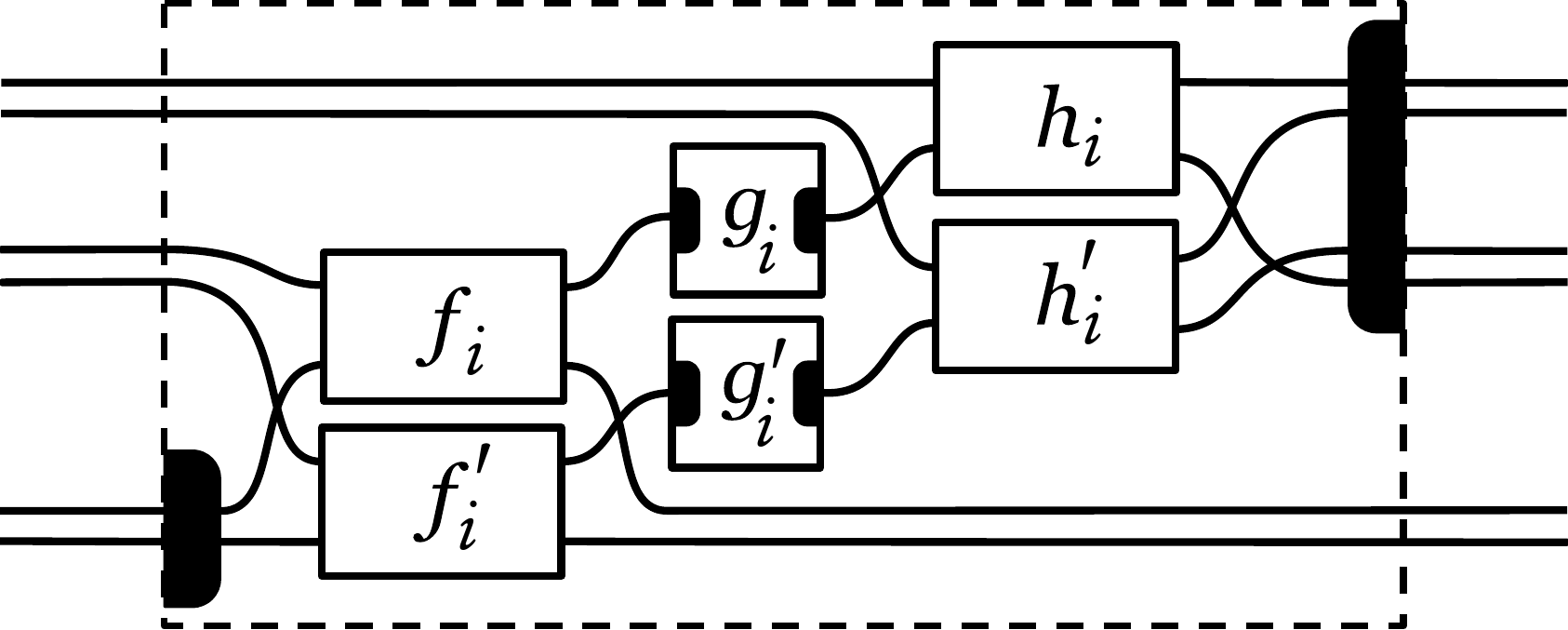}
\medskip
\end{displaymath}
where we make use of the fact that $g_i\tensor g'_i$ belongs to the guarded ideal,
for $g_i$ and $g_i'$ individually do. By induction over $n$ this implies that
the combination of $f$ and $f'$ figuring in~\textbf{(par${}_{\tensor}$)} rule 
has the specified format.

\subsection{Proof of Lemma~\ref{lem:induce}}

Consider a composite of~\eqref{eq:tguard_i} with a diagram of the form
\medskip
\begin{equation}\label{eq:tguard_cmp}
\vcenter{\hbox{\includegraphics[scale=.3]{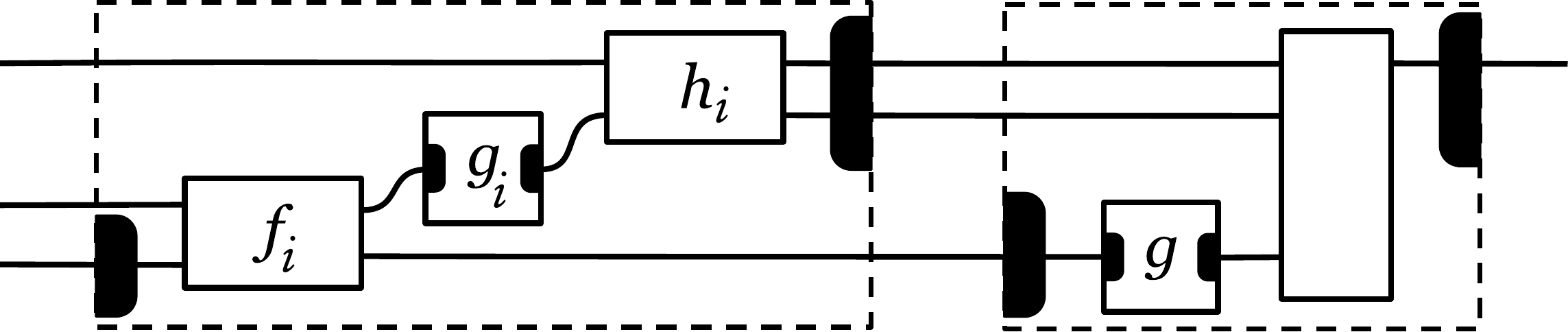}}}
\medskip
\end{equation}
We argue that this composite is equivalent to a diagram of the same form as
on the right of~\eqref{eq:tguard_cmp}. Indeed by the axioms of guarded ideals,
we can replace the tensor product of $g_i$ and $g$ with a single guarded morphism,
and then compose the result with $f_i$ to obtain another guarded morphism, say 
$h\in G(X,Y\tensor Z)$. By assumption, the latter can be represented as 
$e\comp (\hat h\tensor\id)$, i.e.\ in summary we obtain  
\medskip
\begin{displaymath}
\includegraphics[scale=.3]{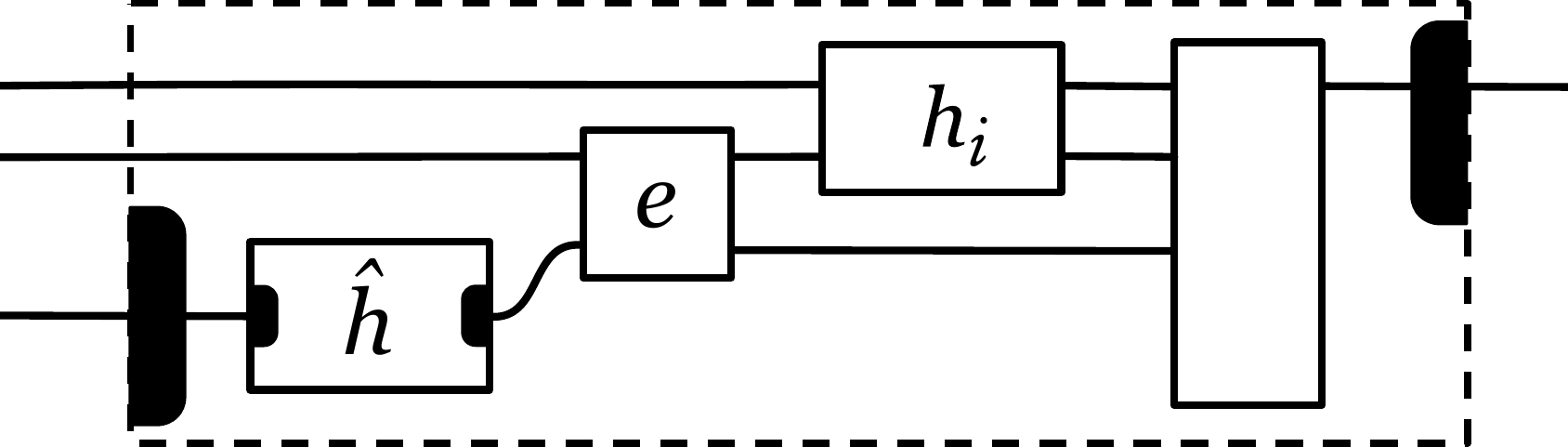}
\medskip
\end{displaymath}
This is clearly reducible to the a diagram in the same form as on the right of~\eqref{eq:tguard_cmp}.

Now, assuming a morphism $f\in\GHom(A\tensor B,C\tensor D)$ as defined in clause~(1) with 
$B=C=I$, note that $(q\tensor\id):(A_n\tensor B_n)\tensor D \to I\tensor D$ falls into
the format specified by the diagram on the right of~\eqref{eq:tguard_cmp} (one takes
$g=\id:I\to I$, which belongs to $G(I,I)$). By inductively applying the above 
argument we contract $f=(q\tensor\id)\comp w_n\comp\ldots w_1\comp (\id\tensor p)$
to the form
\medskip
\begin{displaymath}
\includegraphics[scale=.3]{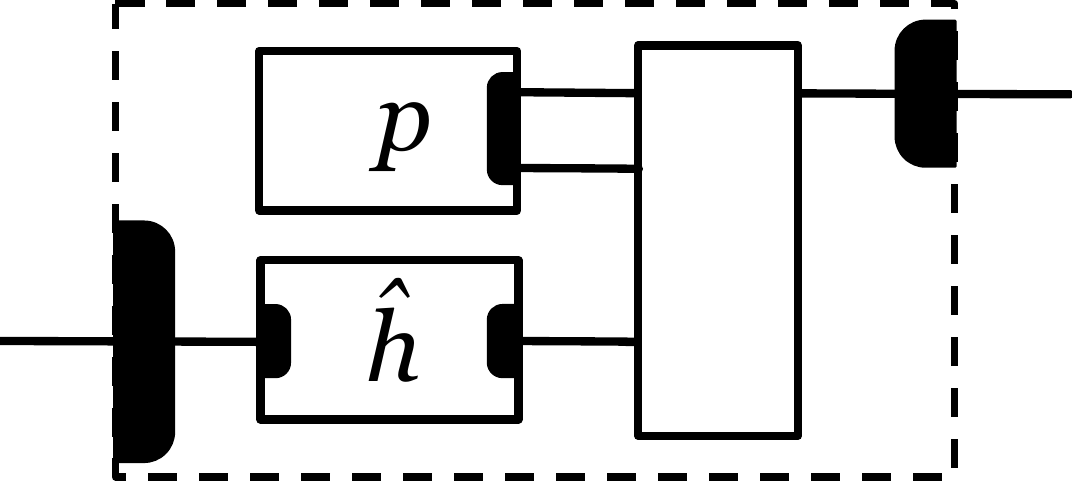}
\medskip
\end{displaymath}
Here $p$ is a guarded morphism because it factors through $\id:I\to I\in G(I,I)$. 
The obtained diagram clearly yields a guarded morphism, and we are done.
\qed

\subsection{Proof of Proposition~\ref{prop:triv-ideal}}

Immediate from Lemma~\ref{lem:ideally-guarded}.\ref{item:induced-less}
below and the assumption that $\BC$ is equipped with the least guarded
structure.

\begin{lem}\label{lem:ideally-guarded}
  Let $(\BC,\tensor,I)$ be a guarded category, with induced guarded
  ideal $G(X,Y)=\IHom(X,Y)$. Then
  \begin{enumerate}
  \item \label{item:induced-less} The guarded structure on $\BC$
    induced by $G$ is contained in the original one.
  \item \label{item:induced} If $\BC$ is ideally guarded, then $G$
    induces the guarded structure of $\BC$.
  \end{enumerate}

\end{lem}

\begin{proof}
  \begin{enumerate}
  \item Immediate from the fact that the guarded structure induced by
    $G$ is the least one containing $G$.
  \item The given guarded ideal inducing the guarded structure of
    $\BC$ is contained in $G$, so the given guarded structure on $\BC$
    is contained in the one induced by
    $\BC$. Part~\ref{item:induced-less} then implies equality. \qed
  \end{enumerate}
\end{proof}
\subsection{Proof of Theorem~\ref{thm:cart-total}}
First of all, note that Lemma~\ref{lem:induce} applies to the case at 
hand, for any $f\in G(X+Y,Z)$ can be represented as follows 
$f = [f\inl,f\inr] = [\id,f\inr]\comp (f\inl+\id)$ where $f\inl\in G(X,Z)$, by 
closure properties of guarded ideals. It remains to prove (1) that the 
\begin{align}\label{eq:co-cart-guard}
	\Hom_{\inr}(X,Y+Z)= \{[\inl,g] h\mid g\in G(W, Y+Z), h:X\to Y+
  W\}
\end{align}
is a correct definition of a guardedness structure, and (2) that it is contained 
in the guardedness structure generated by $G$.
\begin{cenumerate}
 \item It suffices to verify closure under the rules on the right of 
Figure~\ref{fig:co-cart-g}. E.g.\ for~\textbf{(cmp${}_{\mplus}$)}, we have to prove that
$[[\inl,g'] u ,f]\comp [\inl,g] h\in\Hom_{\inr}(X,V+W)$ provided 
$g\in G(Y', Y+Z)$, $h:X\to Y+Y'$, $f:Z\to V+W$, $u:Y\to V+V'$, $g'\in G(V', V+W)$.
Indeed,
\begin{flalign*}
&& [[\inl,g'] u ,f]\comp [\inl,g] h
&  \;= [[\inl,g'] u,[[\inl,g'] u ,f]\comp g] h &~\\
&&&\;= [[\inl,g'],[[\inl,g'] u ,f]\comp g]\comp (u+\id)\comp h \\
&&&\;= [\inl, [g',[[\inl,g'] u ,f]\comp g]]\comp\alpha_{V,V',Y'} (u+\id)\comp h
\end{flalign*}
and the latter is in $\Hom_{\inr}(X,V+W)$ by definition and the fact that 
$[g',[[\inl,g'] u ,f]\comp g]\in G(V'+Y', V+W)$ by axioms of guarded ideals.
  \item By Proposition~\ref{prop:guard_equiv} the general form of a partially guarded morphism 
induced by~\eqref{eq:co-cart-guard} is $[[\inl,g]\comp h, u]:A+B\to C+D$ with $u:B\to C+D$,
$h:A\to C+D'$, $g:D'\to C+D$, which can be structured as follows:
\medskip
\begin{displaymath}
\includegraphics[scale=.3]{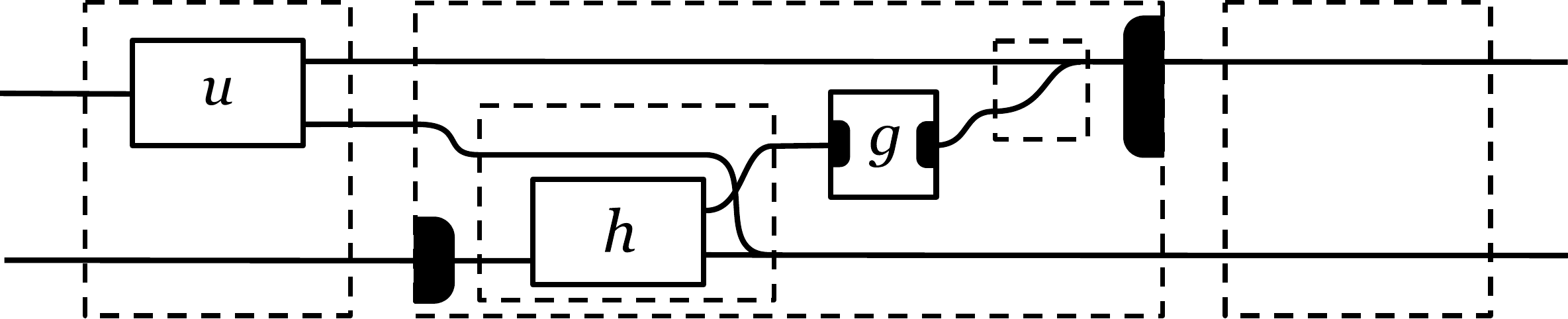}
\medskip
\end{displaymath}
and this indeed fits the format specified by Theorem~\ref{thm:t_to_p}.
\qed

\end{cenumerate}

\subsection{Proof of Corollary~\ref{cor:cart-total2}}

  Record first of all that $G$ is exponential iff
  $f\in G(X\times Y,Z)$ implies $\curry f\in G(X,Z^Y)$, for given
  $f\in G(X\times Y,Z)$, $\curry f = f^Y\comp\curry(\id_{X\times Y})$,
  and given $g\in G(X,Y)$, $g^V = \curry (g\comp\ev)$.

  The given construction produces a guarded category only if $G$ is
  exponential, for $f\in G(X\times Y,Z)$ must by weakening imply
  $f:X\times Y\to^{\pr_1} Z$, whence, by definition,
  $\curry f\in G(X,Z^Y)$.

  Conversely, suppose that $G$ is exponential. We proceed to show that
  the description of the guarded structure on $\BC$ according to
  Corollary~\ref{cor:cart-total1} is equivalent to the current one,
  which will finish the argument. On the one hand, if
  $\curry f\in G(X,Z^Y)$ then
  $f = \ev\comp(\curry f\times\id)= \ev\comp\brks{(\curry
    f)\fst,\snd}$,
  i.e.\ $f$ is $\pr_1$-guarded in the sense of
  Corollary~\ref{cor:cart-total1}; on the other hand, if
  $f=h\brks{g,\snd}$ for some $g\in G(X\times Y,W)$, then
  $\curry f= \curry (h\comp\brks{\ev,\snd})\comp (\curry g)\in
  G(X,Z^Y)$. \qed

\subsection{Proof of Proposition~\ref{prop:mult-guard}}
The axioms of guarded ideals are easy to check. As an example let us verify 
closedness under $\times$: given $f:A\to B$, $g:C\to D$,
\begin{flalign*}
f\oname{next}\times g\oname{next} 
=\brks{f\comp\oname{next}\fst,g\comp\oname{next}\snd}
=\brks{f\comp(\grd\fst),g\comp(\grd\snd)}\oname{next}.&\text{\qed}
\end{flalign*}
\subsection{Proof of Proposition~\ref{prop:ggcompl}}

Induction on $e$. All cases except the one for the trace operation are
analogous to Theorem~\ref{thm:gcompl}. So let $e$ have the form
$\tr_{A,B,C,D}(e')$ where
$e'\in\GHom((A\tensor U)\tensor B,C\tensor (D\tensor U))$. Every path
from an input gate in $A$ to an output gate in $D$ in the diagram of
$\tr_{A,B,C,D}(e')$ is also such a path in $e'$, hence guarded by
induction. The only new loops in the diagram of $\tr_{A,B,C,D}(e')$
are the ones generated by the current application of the trace
operator. Every such loop $\pi$ incorporates a path from an input gate
in $U$ to an output gate in $U$, which is guarded by induction; thus,
$\pi$ itself is guarded.

\subsection{Details for Example~\ref{expl:n-idl-counter} (Right Hand
  Diagram)}

To see that the necessary condition from
Proposition~\ref{prop:ggcompl} holds, note that both the loop
through~$f$ and~$g$ and the path from the unguarded input to the
guarded output of the diagram are guarded. We show that the diagram is
not induced by an expression for which the indicated overall
guardedness typing of the diagram (one unguarded input, one guarded
output) is derivable: The paths connecting the unguarded input and the
guarded output of the diagram with the loop preclude a derivation
using \textbf{(vac${}_{\tensor}$)}; the only way that remains is to
apply the rule for $\tr$. But both ways of cutting the loop (in either
case marking the newly open input gate of the diagram as unguarded and
the new output gate as guarded in order to enable application of
$\tr$) lead to diagrams that have an unguarded path from an unguarded
input to a guarded output, violating the necessary condition from
Proposition~\ref{prop:ggcompl}.

\subsection{Proof of Theorem~\ref{thm:ggcompl-conv}}
Induction on the number of loops in $\Delta$, with
Theorem~\ref{thm:gcompl} (plus the standard fact that,
disregarding guardedness, every acyclic diagram is induced by some
trace-free morphism expression) as the base case. The inductive step
is as follows.

Recall that there are only two types of basic boxes regarding their
decoration, the basic generic guards and boxes with only guarded
inputs and only unguarded outputs. In reference to the colour of the
decorations, we call the former \emph{black} and the latter
\emph{white}.

Let $U$ denote the set of nodes $n$ in $\Delta$ that have an unguarded
path from their inputs to some output gate in $D$ (i.e.\ the unguarded
path includes $n$ itself); dually, let $V$ denote the set of nodes in
$\Delta$ that have an unguarded path from some input gate in $A$ to
their outputs. By the simplified characterization of guarded paths in
ideally guarded diagrams, all nodes in $U\cup V$ must be white. 

Then the assumption implies that  
\begin{equation}\label{eq:uv-disjoint}
  U\cap V = \emptyset.
\end{equation}
Since we are in the inductive step, there exists a loop $\pi$ in
$\Delta$.  \medskip

\noindent\textbf{Claim 1:}  There is some wire $w$ belonging
to $\pi$ that connects an output gate $O$ of a basic box $f\notin V$
to an input gate $I$ of a basic box $g\notin U$.\medskip

\noindent To see this, assume for a contradiction that $w$ fails to
exist, i.e.\ every wire in~$\pi$ is attached either to an output of a
box in $V$ or to an input of a box in $U$. Pick some wire $v$ on
$\pi$, and assume w.l.o.g.\ that $v$ is attached to an input of a box
in~$U$. Then by~\eqref{eq:uv-disjoint}, the same must hold for the
next wire on $\pi$. Continuing around the loop, we find that all boxes
on $\pi$ are in $U$, in particular are white, contradicting the
assumption that $\pi$ is guarded. This proves Claim~1.\medskip

\noindent Now take $w$ as in Claim~1. Briefly, we can cut $w$, apply
the inductive assumption and then reintroduce~$w$ by means of the
trace operator. In detail, let the diagram~$\Delta'$ arise from
$\Delta$ by cutting $w$, let $A'$ consist of the gates in $A$ and the
newly open input gate~$I$, and let $D'$ consist of the gates in $D$
and the newly open output gate $O$. Now since $f\notin V$ and
$g\notin U$, every path $\pi'$ from an input gate in~$A'$ to an output
gate in $D'$ falls within one of the following cases.
\begin{itemize}
\item $\pi'$ runs from a gate in $A$ to a gate in $D$. Since $\pi'$ is
  already present in $\Delta$, $\pi'$ is then guarded by assumption.
\item $\pi'$ runs from $I$ to $O$. Then the nodes of $\pi'$ form a loop
  in $\Delta$, so that $\pi'$ is guarded by assumption.
\item $\pi'$ runs from $I$ to a gate in $D$. Since $g\notin U$, $\pi'$
  is guarded.
\item Dually, $\pi'$ is guarded if it runs from a gate in $A$ to $O$.
\end{itemize}
Finally, all loops in $\Delta'$ are already present in $\Delta$, hence
guarded by assumption. 

By the inductive hypothesis, we therefore have
$e\in\GHom(A'\tensor B,C\tensor D')=\GHom((A\tensor U)\tensor
B,C\tensor (D\tensor U))$
inducing $\Delta'$, where $U$ is the joint type of $I$ and $O$. Then,
$\Delta$ is induced the expression $\tr(e')$.\qed

\subsection{Proof of Lemma~\ref{lem:iter}}
\begin{cenumerate}
\item Suppose, $(\argument)^\istar$ is uniform. Given, $f:X\to_{\inr} Y+X$,
let 
\begin{align*}
w = [\inl (\id+\inr) f, (\inl+\inl) f] : X+X\to_{\inr} (Y+(X+X))+(X+X)).
\end{align*}
We are going to show that $([\inl,f] f)^\istar=w^{\istar\istar}\inr$
and $([\id,\inr] w)^\istar\inr = f^\istar$ which implies the identity in question
by the codiagonal axiom. On the one hand,
\begin{flalign*}
&&w^\istar\inr 
=&\; [\id,w^\istar] w\inr &\by{fixpoint}\\
&&=&\; [\id,w^\istar] (\inl+\inl) f\\
&&=&\; [\inl,w^\istar\inl] f\\
&&=&\; [\inl,[\id,w^\istar] w\inl] f&\by{fixpoint}\\
&&=&\; [\inl,[\id,w^\istar] \inl  (\id+\inr) f] f\\
&&=&\; [\inl, (\id+\inr) f] f\\
&&=&\; (\id+\inr) [\inl, f] f.
\end{flalign*}
By uniformity this implies $w^{\istar\istar}\inr = ([\inl, f] f)^\istar$.
To show $([\id,\inr] w)^\istar\inr = f^\istar$, observe that
\begin{align*}
([\id,\inr] w)^\istar\inr 
=&\; [(\id+\inr) f, [\inl,\inr\inl] f]^\istar \inr\\
=&\; [(\id+\inr) f, (\id+\inl) f]^\istar \inr.
\end{align*}
Since 
\begin{align*}
(\id+\nabla)\comp [(\id+\inr) f, (\id+\inl) f]=[f,f]=f\nabla,
\end{align*}
by uniformity, $[(\id+\inr) f, (\id+\inl) f]^\istar = f^\istar\nabla$,
and therefore $([\id,\inr] w)^\istar\inr = f^\istar\nabla\inr = f^\istar$.
\item Assume that $(\argument)^{\istar}$ is squarable
and uniform w.r.t.\ coproduct injections. Consider $g : X \to_{\inr} Y + Z$ and $f:Z\to Y+X$ 
(and omit the analogous symmetric option $g : X \to Y + Z$ and $f:Z\to_{\inr} Y+X$). 
We introduce
\begin{align*}
h = [(\id + \inr)\comp g, (\id + \inl) \comp f] : X + Z \to_{\inr} (Y+(X+Z))
\end{align*}
where the guardedness annotation is provable by~\textbf{(cmp${}_{\mplus}$)} 
and~\textbf{(par${}_{\mplus}$)} 
and apply the squaring identity to it. Thus, 
$h^\istar = [[\inl, (\id+\inl)\comp f]\comp g, [\inl, \comp (\id+\inr)\comp g]\comp f]^\istar$. Since 
$[[\inl, (\id+\inl)\comp f]\comp g, [\inl, \comp (\id+\inr)\comp g]\comp f]\inl = (\id+\inl)\comp[\inl, f]\comp g$, by uniformity 
this implies $h^\istar\inl = ([\inl, f]\comp g)^\istar$ and, analogously, 
 $h^\istar\inr = ([\inl, g]\comp f)^\istar$. Now,
\begin{align*}
 ([\inl, f]\comp g)^\istar = h^\istar\inl = [\id,h^\istar]\comp h\inl = [\id,h^\istar\inr] g = [\id, ([\inl, g]\comp f)^\istar] g,
\end{align*}
and we are done. 
 \item Assume that $(\argument)^{\istar}$ is Conway. First, we show uniformity w.r.t.\
isomorphisms. Let $f:X\to_{\inl} Y+X$ and let for some $i:X\to X'$, $j:X'\to X$,
$i\comp j =\id$, $j\comp i = \id$, $(\id + i)\comp f = g\comp i$. Then
\begin{flalign*}
&&f^\istar 
=&\; (f\comp j\comp i)^\istar\\
&&=&\; ((\id+i)\comp f\comp j)^\istar\comp i&\by{dinaturality}\\
&&=&\; (g\comp i\comp j)^\istar\comp i\\
&&=&\; g^\istar\comp i.
\end{flalign*} 
We proceed with the proof of the general case, and now we can stick w.l.o.g.\
to the coproduct injections of the form $\inl:X\to X+Y$. Suppose that for some 
$f:X\to_{\inl} Z+X$ and $h:X+Y\to_{\inl} Z+(X+Y)$, $(\id+\inl)\comp f = h\comp\inl$ and show that
$f^\istar = h^\istar\comp\inl$. Note that by assumption, 
$h = [(\id+\inl)\comp f, g]$ where $g = h\inr$, so we only need to show
\begin{align}\label{eq:inj_uni}
[(\id+\inl)\comp f, g]^\istar\inl = f^\istar.
\end{align}
First we tackle the following special case of~\eqref{eq:inj_uni}:
\begin{align}\label{eq:inj_uni1}
[(\id+\inl)\comp f, (\id+\inl)\comp g]^\istar\inl = f^\istar
\end{align}
We have
\begin{flalign*}
&& [(\id+\inl)\comp f,& (\id+\inl)\comp g]^\istar \inl\\
&&=&\;  ((\id+\inl)\comp [f,g])^\istar \inl &\by{dinaturality}\\
&&=&\;  ([f,g]\inl)^\istar\\
&&=&\;  f^\istar.
\end{flalign*}

Finally, let us show~\eqref{eq:inj_uni} in the general form.
Let $w = [\inl(\id+\inl)f, (\inl +\id)\comp g]: X+Y \to_{\inr} (Z + (X+Y)) + (X+Y)$. It 
is straightforward that $([\id,\inr]\comp w)^\istar = [(\id+\inl)\comp f, g]^\istar$,
so, by the codiagonal identity, we are left to show that $w^{\istar\istar}\inl = f^\istar$.
Indeed,
\begin{flalign*}
&& w^{\istar\istar}\inl
=&\; (([\inl,(\id+\inl) f] +\id)\comp [\inl\inr, (\inl + \id)\comp g])^{\istar\istar}\inl\\
&&=&\; ([\inl,(\id+\inl) f] \comp [\inl\inr, (\inl + \id)\comp g]^{\istar})^{\istar}\inl&\by{naturality}\\
&&=&\; f^\istar. &\by{\eqref{eq:inj_uni1}}
\end{flalign*}
and we are done.\qed
\end{cenumerate}

\subsection{Proof of Theorem~\ref{thm:grec_from_rec}}
We denote the strength of~$T$ by
$\tau_{X,Y}:X\times TY\to T(X\times Y)$, and its transpose
$(T\gamma_{Y,X})\comp\tau_{Y,X}\comp\gamma_{TX,Y}$ by
$\hat\tau_{X,Y}:TX\times Y\to T(X\times Y)$.

We need the following technical fact. 
\begin{lem}\label{lem:par_ext}
  Let $\BC$ be a Cartesian category, and let $\BBT$ be a strong
  monad on~$\BC$.  Suppose $(\argument)_{\iistar}$ is a recursion
  operator on $\BC^\BBT_{\star}$ satisfying
  naturality. Then~$(\argument)_{\iistar}$ extends to morphisms of the
  form $f:V\times A\to A$ with $(A,\algebra)\in |\BC^\BBT_\star|$ and
  $V\in |\BC|$. The extended operator
  satisfies
  $f_{\iistar} = (\algebra (Tf)\comp\hat\tau:TV\times A\to A)_{\iistar}\comp\eta$ for $f:V\times A\to A$.
\end{lem}
\begin{proof}
Given a monad algebra $(A,\algebra)$ for $T$ and $f:V\times A\to A$ as in the claim,  put $f_{\iistar} = (\algebra\comp (Tf)\comp\hat\tau)_{\iistar}\comp\eta$,
where the application of the original~$(\argument)_{\iistar}$ on the right hand side is defined  because it involves the free algebra $(TV,\mu)$
instead of $V$. We have to check that this definition agrees with the original one on $\BC^\BBT_\star$.
So let $f:B\times A\to A$ with $(B,\beta)\in |\BC^\BBT_\star|$. Then 
\begin{align*}
(\algebra\comp (Tf)\comp\hat\tau)_{\iistar}\comp\eta
&\;= (\algebra\comp (Tf)\comp\hat\tau\comp (\eta\times\id))_{\iistar}\\
&\;= (\algebra\comp (Tf)\comp\eta)_{\iistar}\\
&\;= (\algebra\eta\comp f)_{\iistar}\\
&\;= f_{\iistar}.
\end{align*}
This argument also shows that the extended operator satisfies $f_{\iistar} = (\algebra (Tf)\comp\hat\tau)_{\iistar}\comp\eta$.
\qed
\end{proof}

\noindent The proof of Theorem~\ref{thm:grec_from_rec} then proceeds
as follows.  Let us first check that the definition of
$(\argument)_{\istar}$ via $(\argument)_{\iistar}$ does not depend on
the factorization of $f$ as $g\comp (\id\times\eta)$.  Suppose
$f=g'\comp (\id\times\eta)=g\comp (\id\times\eta)$. Then
\begin{flalign*}
&&g\brks{\id,(\eta g)_{\iistar}}
=&\; g\brks{\id,(\eta g\comp\brks{\fst,\eta g})_{\iistar}} &\by{squaring}\\
&&=&\; g\brks{\id,(\eta g'\comp\brks{\fst,\eta g})_{\iistar}}&\by{assumption}\\
&&=&\; g\brks{\id,\eta g'\comp\brks{\id,(\eta g\comp\brks{\fst,\eta g'})_{\iistar}}}&\by{dinaturality}\\
&&=&\; g'\brks{\id,\eta g'\comp\brks{\id,(\eta g'\comp\brks{\fst,\eta g'})_{\iistar}}}\qquad& \by{assumption (twice)}\\
&&=&\; g'\brks{\id,\eta g'\comp\brks{\id,(\eta g')_{\iistar}}}&\by{squaring}\\
&&=&\; g'\brks{\id,(\eta g')_{\iistar}}&\by{fixpoint}
\end{flalign*}

Let us check that the mutual transformations between $(\argument)_\iistar$ and 
$(\argument)_\istar$ are mutually inverse.
\begin{flalign*}
 \intertext{$(\argument)_\iistar\to (\argument)_\istar\to (\argument)_\iistar$: given $f:B\times A\to A$ with $(A,\algebra)\in |\BC^{\BBT}|$,}
&&\algebra\eta\comp f\comp  (\id\times\algebra(T&\algebra))\comp\brks{\id,(\eta\eta\comp f\comp (\id\times\algebra (T\algebra)))_{\iistar}}\\ 
&&=&\;f\comp  \brks{\id,\algebra(T\algebra)\comp(\eta\eta\comp f\comp (\id\times\algebra (T\algebra)))_{\iistar}}\\ 
&&=&\;f\comp  \brks{\id,\algebra(T\algebra)\eta\eta\comp(f\comp (\id\times\algebra (T\algebra)\eta\eta))_{\iistar}} &\by{dinaturality}\\
&&=&\;f\comp  \brks{\id, f_{\iistar}}\\
&&=&\; f_{\iistar}.&\by{fixpoint}
\intertext{$(\argument)_\istar\to (\argument)_\iistar\to (\argument)_\istar$: given $f=g\comp (\id\times\eta):Y\times X\to X$,}
&&g\brks{\id,\mu\comp(\eta\eta g(\id&\times\mu))_{\istar}}\\
&&=&\;g\brks{\id,\mu\eta\eta\comp(g(\id\times\mu\eta\eta))_{\istar}} &\by{dinaturality}\\
&&=&\;g\brks{\id,\eta\comp(g\comp (\id\times\eta))_{\istar}}\\
&&=&\;(g\comp (\id\times\eta))_{\istar}&\by{fixpoint}\\
&&=&\;f_\istar.
\end{flalign*}
Next, let us verify that the properties of fixpoints transfer along the transitions
$(\argument)_\iistar\to (\argument)_\istar$ and $(\argument)_\istar\to (\argument)_\iistar$.
It suffices to handle the fixpoint, naturality, squaring, and diagonal laws. Consider the transition
$(\argument)_\istar\to (\argument)_\iistar$.
\begin{flalign*}
\intertext{\emph{fixpoint:}}
&&f_{\iistar}
=&\; \algebra\comp(\eta f(\id\times\algebra))_{\istar}\\
&&=&\; f(\id\times\algebra)\brks{\id,(\eta f(\id\times\algebra))_{\istar}}&\by{fixpoint}\\
&&=&\; f\brks{\id,\algebra\comp(\eta f(\id\times\algebra))_{\istar}}\\
&&=&\; f\brks{\id,f_\istar}\\
\intertext{\emph{naturality:}}
&&f_{\iistar}\comp g
=&\; \algebra\comp(\eta f(\id\times\algebra))_{\istar}\comp g\\
&&=&\; \algebra\comp(\eta f(g\times\algebra))_{\istar}&\by{naturality}\\
&&=&\; (f\comp (g\times\id))_{\iistar}
\intertext{\emph{squaring:}}
&&f_{\iistar}
=&\; \algebra\comp(\eta f(\id\times\algebra))_{\istar}\\
&&=&\; \algebra\comp(\eta f(\id\times\algebra))_{\istar}\\
&&=&\; \algebra\comp(\eta f(\id\times\algebra)\comp\brks{\fst,\eta f(\id\times\algebra)})_{\istar}&\by{squaring}\\
&&=&\; \algebra\comp(\eta f\comp\brks{\fst,f(\id\times\algebra)})_{\istar}\\
&&=&\; \algebra\comp(\eta f\comp\brks{\fst, f}\comp(\id\times\algebra))_{\istar}\\
&&=&\; (f\comp\brks{\fst, f})_{\iistar}.
\intertext{\emph{dinaturality:}}
&&(g\comp\brks{\fst,h})_{\iistar} 
=&\; \algebra\comp(\eta g\comp\brks{\fst,h}(\id\times\algebra))_{\istar}\\
&&=&\; \algebra\eta\comp g\comp\brks{\id,(h\comp(\id\times\algebra)\comp\brks{\fst, \eta g})_{\istar}}&\by{dinaturality}\\
&&=&\; g\comp\brks{\id,(h\comp\brks{\fst,g})_{\istar}} \\
&&=&\; g\comp\brks{\id,(\algebra\eta h\comp\brks{\fst,g})_{\istar}} \\
&&=&\; g\comp\brks{\id,\algebra\comp(\eta h\comp\brks{\fst,g}(\id\times\algebra))_{\istar}} &\by{dinaturality}\\
&&=&\; g\comp\brks{\id,(h\comp\brks{\fst,g})_{\iistar}}. 
\intertext{\emph{diagonal:}}
&&(f\comp\brks{\id,\snd})_{\iistar} 
=&\; \algebra\comp(\eta f\comp\brks{\id,\snd}(\id\times\algebra))_{\istar}\\
&&=&\; \algebra\comp(\eta f\comp\brks{(\id\times\algebra),\algebra\snd})_{\istar}\\
&&=&\; \algebra\comp(\eta f\comp((\id\times\algebra)\times\algebra)\comp\brks{\id,\snd})_{\istar}\\
&&=&\; \algebra((\eta f\comp((\id\times\algebra)\times\algebra))_{\istar\istar}&\by{diagonal}\\
&&=&\; \algebra((\eta f\comp(\id\times\algebra))_{\istar}\comp(\id\times\algebra))_{\istar}&\by{naturality}\\
&&=&\; \algebra\comp\eta\algebra((\eta f\comp(\id\times\algebra))_{\istar}\comp(\id\times\algebra\eta\algebra))_{\istar}\\
&&=&\; \algebra\comp(\eta\algebra(\eta f\comp(\id\times\algebra))_{\istar}\comp(\id\times\algebra))_{\istar}&\by{dinaturality}\\
&&=&\;f_{\iistar\iistar}.
\end{flalign*}
Next, consider the transition $(\argument)_\iistar\to (\argument)_\istar$.
\begin{flalign*}
\intertext{\emph{fixpoint:}}
&&(g\comp (\id\times\eta))_{\istar}
=&\; g\brks{\id,(\eta g)_{\iistar}}\\
&&=&\; g\brks{\id,\eta g\brks{\id,(\eta g)_{\iistar}}}&\by{fixpoint}\\
&&=&\; (g\times\eta) \brks{\id,g\brks{\id,(\eta g)_{\iistar}}}\\
&&=&\; (g\times\eta) \brks{\id,(g\comp (\id\times\eta))_{\istar}}\\
\intertext{\emph{naturality:}}
&&(g\comp (\id\times\eta))_{\istar}\comp f
=&\; g\brks{\id,(\eta g)_{\iistar}}\comp f\\
&&=&\; g\brks{f,(\eta g\comp (f\times\id))_{\iistar}}&\by{naturality}\\
&&=&\; g\comp (f\times\id) \brks{\id,(\eta g\comp (f\times\id))_{\iistar}}\\
&&=&\; (g\comp (f\times\eta))_{\istar}
\intertext{\emph{squaring:}}
&&(g\comp (\id\times\eta))_{\istar}
=&\; g\brks{\id,(\eta g)_{\iistar}}\\
&&=&\; g\brks{\id,\eta g\comp\brks{\id,(\eta g)_{\iistar}}}&\by{fixpoint}\\
&&=&\; g\brks{\id,\eta g\comp\brks{\id,(\eta g\comp\brks{\fst,\eta g})_{\iistar}}}&\by{squaring}\\
&&=&\; g\comp\brks{\fst, \eta g}\brks{\id,(\eta g\comp\brks{\fst, \eta g})_{\iistar}}\\
&&=&\; (g\comp\brks{\fst, \eta g}\comp (\id\times\eta))_{\istar}\\
&&=&\; (g\comp (\id\times\eta)\comp\brks{\fst, g\comp (\id\times\eta)})_{\istar}.
\intertext{\emph{dinaturality:}}
&&(g\comp\brks{\fst,h\comp (\id&\times\eta)})_{\istar}\\ 
&&=&\; (g\comp\brks{\fst,h}\comp  (\id\times\eta))_{\istar} \\
&&=&\; g\comp\brks{\fst,h}\brks{\id,(\eta g\comp\brks{\fst,h})_{\iistar}}\\
&&=&\; g\comp\brks{\id,h\brks{\id,((Tg)\comp\tau\comp\brks{\fst,\eta h})_{\iistar}}}\\
&&=&\; g\comp\brks{\id,h\brks{\id,(Tg)\comp\tau\comp\brks{\id, (\eta h\comp\brks{\fst,(Tg)\comp\tau})_{\iistar}}}}\quad &\by{dinaturality}\\
&&=&\; g\comp\brks{\id, h\comp  \brks{\fst, (Tg)\comp\tau}\comp (\eta h\comp  \brks{\fst, (Tg)\comp\tau})_{\iistar}} \\
&&=&\; g\comp\brks{\id, (h\comp  \brks{\fst, (Tg)\comp\tau}\comp (\id\times\eta))_{\istar}} \\
&&=&\; g\comp\brks{\id, (h\comp  \brks{\fst, \eta g})_{\istar}} \\
&&=&\; g\comp\brks{\id, (h\comp  (\id\times\eta)\comp\brks{\fst, g})_{\istar}}.
\end{flalign*}
(the case of $g\comp (\id\times\eta)$ and $h$ instead of $g$ and $h\comp (\id\times\eta)$
is analogous.)

We are left to check the diagonal identity. First, we do it in the non-parametrized
case, i.e.\ for the morphisms of type $g:1\times T(X\times X)\to X$. We identify
such $g$ for brevity with $g:T(X\times X)\to X$ with the implied simplification 
of notation therefore. Our goal thus is the equation 
\begin{align*} %
(g\comp\eta\Delta)_\istar = ((g\comp\eta)_\istar)_\istar.
\end{align*}
We transform the left and the right hand sides as follows:
\begin{flalign*}
&&((g\comp\eta)_\istar)_\istar 
=&\; ((g\comp\tau\comp(\id\times\eta))_\istar)_\istar \hspace{25ex}{}\\
&&=&\; (g\comp\tau\comp\brks{\id,(\eta g\comp\tau)_\iistar})_\istar &\by{definition of $(\argument)_{\iistar}$}\\
&&=&\; (g\comp\hat\tau^\klstar\tau\comp\brks{\eta,(\eta g\comp\hat\tau^\klstar\tau\comp (\eta\times\id))_\iistar})_\istar \\
&&=&\; (g\comp\hat\tau^\klstar\tau\comp\brks{\eta,(\eta g\comp\hat\tau^\klstar\tau)_\iistar\comp\eta})_\istar &\by{naturality}\\
&&=&\; (g\comp\hat\tau^\klstar\tau\comp\brks{\id,(\eta g\comp\hat\tau^\klstar\tau)_\iistar}\comp\eta)_\istar \\
&&=&\; (g\comp\hat\tau^\klstar\tau\comp\brks{\id,\eta g(\hat\tau^\klstar\tau\comp (\id\times \eta g))_\iistar}\comp\eta)_\istar  &\by{dinaturality}\\
&&=&\; (g\comp\hat\tau\comp\brks{\id, g(\hat\tau\comp (\id\times g))_\iistar}\comp\eta)_\istar \\
&&=&\; (g\comp\hat\tau\comp(\id\times g)\brks{\id, (\hat\tau\comp (\id\times g))_\iistar}\comp\eta)_\istar \\
&&=&\; (g\comp(\hat\tau\comp (\id\times g))_\iistar\comp\eta)_\istar  &\by{fixpoint}\\
&&=&\; g\comp(\hat\tau\comp (\id\times g))_\iistar(\eta g\comp(\hat\tau\comp (\id\times g))_\iistar)_\iistar &\by{definition of $(\argument)_{\iistar}$}\\
&&=&\; g\comp(\hat\tau\comp (\id\times g))_\iistar\eta g\comp ((\hat\tau\comp (\id\times g))_\iistar\comp \eta g)_\iistar &\by{dinaturality}\\
&&=&\; g\comp(\hat\tau\comp (\eta g\times g))_\iistar\comp ((\hat\tau\comp (\eta g\times g))_\iistar)_\iistar &\by{naturality}\\
&&=&\; g\comp(\eta\comp (g\times g))_\iistar\comp ((\eta ( g\times g))_\iistar)_\iistar \\
&&=&\; g\comp ((\eta ( g\times g))_\iistar)_\iistar&\by{fixpoint}\\[1.5ex]
&&(g\comp\eta\Delta)_\istar 
=&\; g\comp (T\Delta)\comp (\eta g\comp (T\Delta))_\iistar \\
&&=&\; g\comp ((T\Delta)\comp \eta g)_\iistar &\by{dinaturality}\\
&&=&\; g\comp (\eta\comp (g\times g)\comp\Delta)_\iistar.
\end{flalign*}
The computed vales are equal by the diagonal axiom for $(\argument)_{\iistar}$.

To extend this calculation to the parametrized case, we observe that
the whole situation lifts to the co-Kleisli category $\BC_Y$ for the
comonad $(-)\times Y$ for any parametrizing object $Y\in|\BC|$. In
more detail, the monad $\BBT$ lifts to a strong monad $\bar\BBT$ on
$\BC_Y$ by means of the strength, with $\bar Tf=Tf\hat\tau$ and all
other components of the monad structure (unit, multiplication,
strength) arising by precomposition with $\fst$. The monad and
strength laws are checked straightforwardly in the internal languages
of $\BBT$~\cite{Moggi91}. Then a co-Kleisli morphism $X\times Z\to W$
is guarded in $X$ iff the corresponding $\BC$-morphism
$X\times Z\times Y\to W$ is guarded in $X$, using the simplified
description from Proposition~\ref{prop:mult-guard}, which applies here
because $\BBT$ is strong. A recursion operator on $\BC^\BBT_\star$ is
then essentially the same as a recursion operator on the full
subcategory $\BC^{\bar\BBT}_{Y,\star}$ of $\BC^{\bar\BBT}_{Y}$ spanned
by the objects in $\BC^\BBT_\star$ -- the only difference is that the
latter operator has a parameter of type $Y$ hardwired into the base
category, and recursion operators are parametrized to begin
with. (Going in the other direction, the parameter $Y$ can just be
projected out.) The same applies to guarded recursive operators on
$\BC$ and $\BC_Y$, respectively, thanks to the previous observation
that guardedness is the same in $\BC$ and in $\BC_Y$. Equational laws
of recursion operators transfer between $\BC$ and $\BC_Y$ in the
process, as the operators essentially do not change. By the previous
calculations applied to $\BC_Y$, we obtain that $(-)_\dagger$
satisfies the unparametrized diagonal law in $\BC_Y$. But this implies
that $(-)_\dagger$ satisfies the parametrized diagonal law in
$\BC$. We do this last step in more detail: Let
$f:X\times X\times Y \to X$ in $\BC$, guarded in the first two
arguments. Then $f:X\times X\to X$ in $\BC_Y$ (guarded in both
arguments), so by the unparametrized diagonal law in $\BC_Y$, we have
$(f_\dagger)_\dagger=(f\bar\Delta)_\dagger$ in $\BC_Y$, where
$\bar\Delta:X\to X\times X$ is the diagonal in $\BC_Y$. The left hand
side is the same as $(f_\dagger)_\dagger$ in $\BC$. Translating the
right-hand side into $\BC$, we have $\bar\Delta=\Delta\fst$ in $\BC$
where $\Delta:X\to X\times X$ is the diagonal in $\BC$, and the
composite $f\bar\Delta$ in $\BC_Y$ then becomes, expanding the
definition of co-Kleisli composition, $f\langle\Delta\fst,\snd\rangle$
in $\BC$. Again noting that $(-)_\dagger$ is the same in $\BC_Y$ as in
$\BC$, we have shown that
$(f_\dagger)_\dagger=(f\langle\Delta\fst,\snd\rangle)_\dagger=(f(\Delta\times\id))_\dagger$
in $\BC$, which is the (parametrized) diagonal law up to rebracketing
of the product $X\times X\times Y$.  \qed

\subsection{Proof of Proposition~\ref{prop:trivial-trace}}

It is clear that by the assumption in the statement, the trace
operation is well-defined. It remains to check the equational axioms
given in Figure~\ref{fig:gtrace}. We argue diagrammatically, using the
standard coherence result for symmetric monoidal categories. Note
first that in all equations in Figure~\ref{fig:gtrace}, the two sides
constitute isomorphic diagrams (abstracting away the dotted boxes, and
recalling that the dotted line in Vanishing I connects gates of type
$I$ and hence is, for diagrammatic purposes, not really there). Now
replace all basic (i.e.\ solid) boxes in the diagrams by the pattern
for vacuous guardedness,
\begin{center}
  \includegraphics[scale=.3]{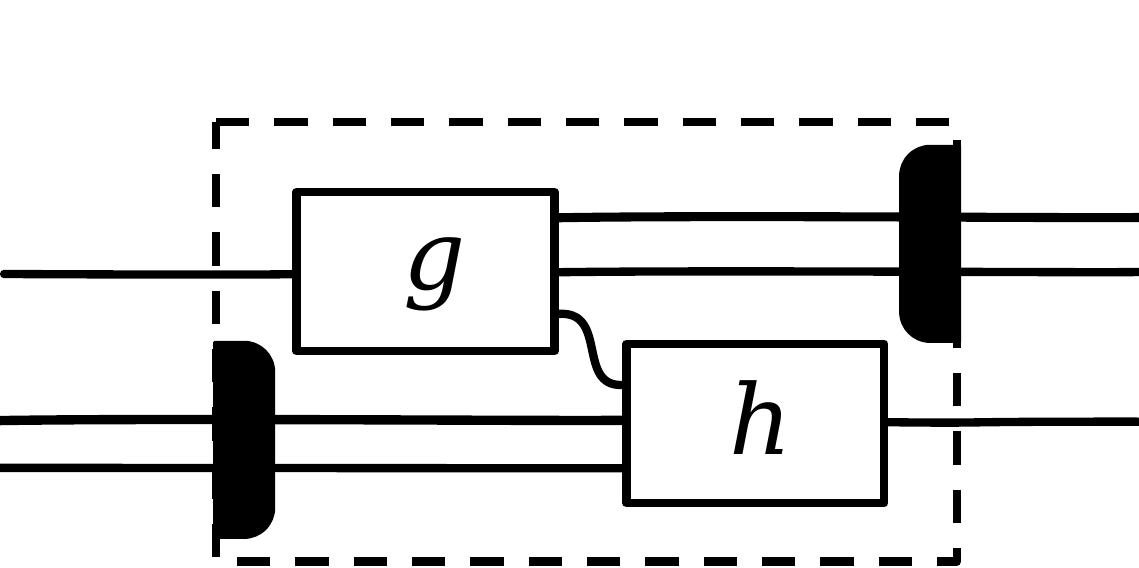}
\end{center}
(or simplified versions thereof with fewer gates, to adjust to the
number of gates of the corresponding basic box in the original
diagram). This makes the diagrams acyclic, as the back-edges appearing
in trace formation now no longer close any loops due to the absence of
paths from unguarded (black) input gates to guarded (black) output
gates in the above pattern. Notice next that the diagrammatic
definition of the trace operation
\begin{center}
  \includegraphics[scale=.3]{nuclear_tr.pdf}
\end{center}
just transforms the given diagram into an isomorphic one. Summing up,
after calculating traces in the above manner, the left and right hand
side of every axiom in Figure~\ref{fig:gtrace} are now isomorphic
acyclic diagrams, so the corresponding equations are valid over
symmetric monoidal categories by the standard coherence
theorem~\cite{Selinger11}. \qed

\subsection{Proof of Theorem~\ref{thm:hilbert-schmidt}}
\begin{lem}\label{lem:comp-hilb}
Let $g\in\HS(A, C)$, $f\in\HS(B, C)$. Then, $\brks{g(a),f(b)} = \brks{b\tensor\theta(g^\dagger)\comma \theta(f)\tensor a}$,
for any $a\in A$, $b\in B$.
\end{lem}
\begin{proof}
\begin{align*}
\brks{g(a),f(b)} 
=&\;\brks{b, f^\dagger(g(a))}\\
=&\;\brks{b\tensor 1, (\id_B\tensor(\theta(g^\dagger))^\dagger)(\theta(f)\tensor \id_A)(1,a)}\\
=&\;\brks{b\tensor\theta(g^\dagger), \theta(f)\tensor a}.
\end{align*}
\qed
\end{proof}
We proceed with the proof of the theorem.
  \emph{\ref{item:hs2}}.: By Remark~\ref{rem:triv-ideal}, $f:A\to D$
  belongs to the guarded ideal induced by the vacuously guarded
  structure iff
  $\upsilon^{-1}_Df\hat\upsilon_A:A\tensor I\to I\tensor D$ is of the
  form
  \begin{align*}
    A\tensor I\xto{\id_A\tensor g} A\tensor E\tensor D\xto{h^\dagger\tensor\id_D} I\tensor D
  \end{align*}
  for suitable $g:I\to E\tensor D$ and $h:I\to A\tensor E$ (eliding
  the associativity isomorphism). By the compactness property of
  nuclear ideals (Definition~\ref{def:nucl}), the condition
  $f\in\HS^2(A,D)$ is characterized by the same condition with $g$
  replaced with a morphism of the form $\theta(g')$ and $h$ replaced
  with a morphism of the form $(\theta(h'))^\dagger$ (for
  Hilbert-Schmidt operators $g',h'$). Since both $\theta$ and
  $(-)^\dagger$ are bijective, the two conditions are equivalent as
  claimed.

  Next we have to establish the
  well-definedness condition from
  Proposition~\ref{prop:trivial-trace}.  To that end, first let us
  argue that a bounded linear operator $f:A\tensor B\to C$ is
  determined by its values on arguments of the form $a\tensor b$ where
  $a\in A$, $b\in B$. Indeed,  we have
  \begin{equation*}
    f^\dagger(c)(a,b) = \brks{a\tensor b, f^\dagger(c)} = \brks{f(a\tensor b),c}%
  \end{equation*}
  for all $c\in C$.%

  Now let $f\in\GHom(A\tensor U)\tensor B, C\tensor (D\tensor U))$,
  with factorization
  $f=(h\tensor\id_{D\tensor U})(\id_{A\tensor U}\tensor g)$ (eliding
  associativity), $g:B\to E\tensor D\tensor U$,
  $h:A \tensor U\tensor E\to C$, as per
  Lemma~\ref{lem:triv}. Write $w$ for the result of calculating
  $\tr_{A,B,C,D}^U(f)$ w.r.t.\ this factorization according
  to~\eqref{eq:trivial-trace}. By the above argument we need to check
  that the values $w(a\tensor b)(c,d)$, for $a\in A$, $b\in B$,
  $c\in C$ and $d\in D$, depend only on $f$.

  By Proposition~\ref{prop:hs-nucl}, there are $g_b:D\tensor U\to E$
  and $h_c:A\tensor U\to E$ such that $g(b) = \theta(g_b^\dagger)$,
  $h(c) = \theta(h_c)$ (where we identify elements of a space $X$ with
  linear operators $\real\to X$). Moreover, let
  $\hat h:C\tensor A\to U\tensor E$ and
  $\hat g:B\tensor D\to U\tensor E$ be defined by
  $\hat h(c\tensor a)(u,e) = h(c)(a,u,e)$,
  $\hat g(b\tensor d)(u,e) = g(b)(e,d,u)$.
Let $u:A\tensor E\tensor D\tensor U \to A\tensor U\tensor E\tensor D$
be the permutation isomorphism involved in \eqref{eq:trivial-trace}.
Then
\begin{flalign*}
&&w(a\tensor b)(c,d) 
=&\; \brks{(h^\dagger\tensor\id_{D})\comp u\comp (a\tensor g(b)), c\tensor d}\\
&&=&\; \brks{u\comp(a\tensor g(b)), h(c)\tensor d}&\by{defn.~of~$\dagger$}\\
&&=&\; \brks{\hat g(b\tensor d), \hat h(c\tensor a)} &\by{Lemma~\ref{lem:comp-hilb}}\\
&&=&\;\sum\nolimits_{i,j} {h}(c)(a, u_j, e_i)\comp {g}(b)(e_i, d, u_j)\\
&&=&\;\sum\nolimits_{i,j} \brks{g_b(d\tensor u_i),e_j}\brks{e_j, h_c(a\tensor u_i)}\\
&&=&\;\sum\nolimits_{i} \brks{g_b(d\tensor u_i), h_c(a\tensor u_i)}\\
&&=&\;\sum\nolimits_{i} \brks{a\tensor u_i\tensor \theta(g_b^\dagger), \theta(h_c)\tensor d\tensor u_i}&\by{Lemma~\ref{lem:comp-hilb}}\\
&&=&\;\sum\nolimits_{i} \brks{a\tensor u_i\tensor g(b), h(c)\tensor d\tensor u_i}\\
&&=&\;\sum\nolimits_{i} \brks{a\tensor u_i, h(c)}\comp\brks{g(b), d\tensor u_i}\\
&&=&\;\sum\nolimits_{i} \brks{h^\dagger(a\tensor u_i), c}\comp\brks{g(b), d\tensor u_i}\\
&&=&\;\sum\nolimits_{i} \brks{h^\dagger(a\tensor u_i)\tensor g(b), c\tensor d\tensor u_i}\\
&&=&\;\sum\nolimits_{i} \brks{(h^\dagger\tensor\id_{D\tensor U})\comp (a\tensor u_i\tensor g(b)), c\tensor d\tensor u_i}\\
&&=&\;\sum\nolimits_{i} \brks{f(a\tensor u_i\tensor b), c\tensor d\tensor u_i}
\end{flalign*}
depends only on $f$, as required.

\emph{\ref{item:hilb-trace-dagger}.}: Since $(-)^\dagger$ preserves
the monoidal structure, its combination with the symmetry as in the
statement can be seen as just realizing the duality discussed in
Remark~\ref{rem:dual}. In particular, the given factorization of $f$
witnessing the guardedness typing assumed in the statement induces a
corresponding factorization of
$\gamma_{B,A\tensor U} f^\dagger \gamma_{D\tensor U,C}$, so that
indeed
$\gamma_{B,A\tensor U} f^\dagger \gamma_{D\tensor U,C}\in
\GHom((D\tensor U)\tensor C,B\tensor (A\tensor U))$;
the same observation implies the claimed equality. \qed

\end{document}

